\documentclass[11pt]{article}

\usepackage[letterpaper,margin=0.90in]{geometry}
\usepackage{amsmath, amssymb, amsthm, amsfonts}
\usepackage{bbm}

\usepackage{ifthen}
\usepackage{tikz}
\usetikzlibrary{positioning,decorations.pathreplacing}

\usepackage{cite}
\usepackage{appendix}
\usepackage{graphicx}
\usepackage{color}
\usepackage[linesnumbered,vlined]{algorithm2e}
\usepackage[noend]{algpseudocode}
\usepackage{epstopdf}
\usepackage{wrapfig}
\usepackage{paralist}
\usepackage[textsize=tiny]{todonotes}

\usepackage{framed}
\usepackage[framemethod=tikz]{mdframed}
\usepackage[bottom]{footmisc}
\usepackage{enumitem}
\setitemize{noitemsep,topsep=3pt,parsep=3pt,partopsep=3pt}
\usepackage[font=small]{caption}
\usepackage{xspace}

\newtheorem{theorem}{Theorem}[section]
\newtheorem{lemma}[theorem]{Lemma}
\newtheorem{meta-theorem}[theorem]{Meta-Theorem}

\newtheorem{corollary}[theorem]{Corollary}

\newtheorem{definition}[theorem]{Definition}
\newtheorem{fact}[theorem]{Fact}
\usepackage{wrapfig}

\definecolor{darkgreen}{rgb}{0,0.5,0}
\usepackage{hyperref}
\hypersetup{
    unicode=false,          
    colorlinks=true,        
    linkcolor=red,          
    citecolor=darkgreen,    
    filecolor=magenta,      
    urlcolor=darkblue       
}

\definecolor{darkgreen}{rgb}{0,0.5,0}
\definecolor{ddarkgreen}{rgb}{0,0.35,0}
\definecolor{darkblue}{rgb}{0.0,0.0,0.5}
\definecolor{ddarkblue}{rgb}{0.0,0.0,0.4}
\definecolor{darkred}{rgb}{0.45,0.0,0.0}
\definecolor{ddarkred}{rgb}{0.35,0.0,0.0}
\usepackage[capitalize,nameinlink]{cleveref}
\crefname{theorem}{Theorem}{Theorems}
\Crefname{lemma}{Lemma}{Lemmas}

\algnewcommand\algorithmicswitch{\textbf{switch}}
\algnewcommand\algorithmiccase{\textbf{case}}

\algdef{SE}[SWITCH]{Switch}{EndSwitch}[1]{\algorithmicswitch\ #1\ \algorithmicdo}{\algorithmicend\ \algorithmicswitch}%
\algdef{SE}[CASE]{Case}{EndCase}[1]{\algorithmiccase\ #1}{\algorithmicend\ \algorithmiccase}%
\algtext*{EndSwitch}%
\algtext*{EndCase}%

\newcommand{\set}[1]{\left\{ #1 \right\}}

\newcommand{\Hcal}{\mathcal{H}}
\newcommand{\ceil}[1]{\lceil #1 \rceil}

\newcommand{\ybar}{Y}
\newcommand{\rhobar}{\bar{\rho}}

\newcommand{\E}{\mathbb{E}}
\newcommand{\poly}{\text{poly}}

\newenvironment{theorem-repeat}[1]{\begin{trivlist}
		\item[\hspace{\labelsep}{\bf\noindent Theorem \ref{#1} }]\em }%
	{\end{trivlist}}

\newcommand*\samethanks[1][\value{footnote}]{\footnotemark[#1]}

\begin{document}
\begin{titlepage}
\date{}

\title{Derandomizing Local Distributed Algorithms under \\ Bandwidth Restrictions}

\author{
 Keren Censor-Hillel\thanks{Supported in part by the Israel Science Foundation (grant 1696/14).}\\
  \small Technion \\
  \small ckeren@cs.technion.ac.il
	\and
	Merav Parter\\
	\small MIT \\
	\small parter@mit.edu
 \and
 Gregory Schwartzman\samethanks[1]\\
  \small Technion \\
  \small gregory.schwartzman@gmail.com
}

\maketitle

\begin{abstract}
This paper addresses the cornerstone family of \emph{local problems} in distributed computing, and investigates the curious gap between randomized and deterministic solutions under bandwidth restrictions.

Our main contribution is in providing tools for derandomizing solutions to local problems, when the $n$ nodes can only send $O(\log n)$-bit messages in each round of communication. We combine bounded independence, which we show to be sufficient for some algorithms, with the method of conditional expectations and with additional machinery, to obtain the following results.

Our techniques give a deterministic maximal independent set (MIS) algorithm in the CONGEST model, where the communication graph is identical to the input graph, in $O(D\log^2 n)$ rounds, where $D$ is the diameter of the graph. The best known running time in terms of $n$ alone is $2^{O(\sqrt{\log n})}$, which is super-polylogarithmic, and requires large messages. For the CONGEST model, the only known previous solution is a coloring-based $O(\Delta + \log^* n)$-round algorithm, where $\Delta$ is the maximal degree in the graph. To the best of our knowledge, ours is the first deterministic MIS algorithm for the CONGEST model, which for polylogarithmic values of $D$ is only a polylogarithmic factor off compared with its randomized counterparts.

On the way to obtaining the above, we show that in the \emph{Congested Clique} model, which allows all-to-all communication, there is a deterministic MIS algorithm that runs in $O(\log \Delta \log n)$ rounds.
When $\Delta=O(n^{1/3})$, the bound improves to $O(\log \Delta)$ and holds also for $(\Delta+1)$-coloring.

In addition, we deterministically construct a $(2k-1)$-spanner with $O(kn^{1+1/k}\log n)$ edges in $O(k \log n)$ rounds. For comparison, in the more stringent CONGEST model, the best deterministic algorithm for constructing a $(2k-1)$-spanner with $O(kn^{1+1/k})$ edges runs in $O(n^{1-1/k})$ rounds.
\end{abstract}

\thispagestyle{empty}
\end{titlepage}

\section{Introduction}
\subsection{Motivation}
A cornerstone family of problems in distributed computing are the so-called \emph{local problems}. These include finding a maximal independent set (MIS), a $(\Delta+1)$-coloring where $\Delta$ is the maximal degree in the network graph, finding a maximal matching, constructing multiplicative spanners, and more. Intuitively, as opposed to \emph{global problems}, local problems admit solutions that do not require communication over the entire network graph.

One fundamental characteristic of distributed algorithms for local problems is whether they are deterministic or randomized. Currently, there exists a curious gap between the known complexities of randomized and deterministic solutions for local problems. Interestingly, the main indistinguishability-based technique used for obtaining the relatively few lower bounds that are known seems unsuitable for separating these cases.
Building upon an important new lower bound technique of Brandt et al.~\cite{BrandtFHKLRSU16}, a beautiful recent work of Chang et al.~\cite{ChangKP16} sheds some light over this question, by proving that the randomized complexity of any local problem is at least its deterministic complexity on instances of size $\sqrt{\log n}$. In addition, they show an exponential separation between the randomized and deterministic complexity of $\Delta$-coloring trees. These results hold in the \emph{LOCAL} model, which allows unbounded messages.

The size of messages that are sent throughout the computation is indeed a second major attribute of distributed algorithms. Another central model for distributed computing is the \emph{CONGEST} model, in which message sizes are restricted, typically to $O(\log n)$ bits. For global problems, such as computing a minimum spanning tree (MST) and more, separations between LOCAL and CONGEST are known~\cite{Elkin-2004,Peleg-Rubinovich-1999,DasSarma-11}. Intriguingly, such separations are in general not known for local problems and are a central open question. Hence, it is crucial to study the complexity of local problems under bandwidth restrictions. Surprisingly, some (but not all) of the algorithms for local problems already use only small messages. This includes, for example, the classic MIS algorithms of Luby~\cite{luby1986simple} and Cole and Vishkin~\cite{ColeV86}, as well as the deterministic coloring algorithm of Barenboim~\cite{Barenboim15}.

This paper investigates the complexity of local problems in the domain that lies in the intersection of the randomized-deterministic gaps and the natural assumption of restricted bandwidth. Specifically, we show how to derandomize distributed algorithms for local problems in the CONGEST model and in its relaxation known as the congested clique model, which further allows all-to-all communication (regardless of the input graph for which the problem needs to be solved), thereby abstracting away the effect of distances and focusing on the bandwidth restrictions. The essence of our derandomization technique lies in the design of a \emph{sample} space that is easy to search by communicating only little information to a leader.

\subsection{Our Contribution}
\paragraph{Maximal Independent Set (MIS):} We begin by derandomizing the MIS algorithm of Ghaffari~\cite{ghaffari2016improved}, which runs in $O(\log\Delta)+2^{O(\sqrt{\log\log n})}$ rounds, w.h.p\footnote{As standard in this area, \emph{with high probability} means with probability that is at least $1-1/n$.}. In a nutshell, this algorithm works in constant-round phases, in which nodes choose to mark themselves with probabilities that evolve depending on the previous probabilities of neighbors. A marked node that does not have any marked neighbors joins the MIS and all of its neighbors remove themselves from the graph. The analysis shows that after $O(\log\Delta)$ phases the graph consists of a convenient decomposition into small clusters for which the problem can be solved fast. This is called the shattering phenomena, and it appears in several distributed algorithms for local problems (see, e.g.,~\cite{barenboim2012locality}).

We first show that a tighter analysis for the congested clique model of Ghaffari's MIS algorithm can improve its running time from $O(\log\Delta+\log^* n)$ (which follows from combining~\cite{ghaffari2016improved} with the new connectivity result of ~\cite{Ghaffari-Parter-MST}) to $O(\log \Delta)$ rounds.
\newcommand{\TheoremMISlogDelta}
{
There is a randomized algorithm that computes MIS in the congested clique model within $O(\log \Delta)$ rounds with high probability.
}
\begin{theorem}
\label{theorem:randcongestedclique}
\TheoremMISlogDelta
\end{theorem}

For the derandomization, we use the method of conditional expectations (see e.g.,~\cite[Chapter 6.3]{UpfalMBook}). In our context, this shows the existence of an assignment to the random choices made by the nodes that attains the desired property of removing a sufficiently large part of the graph in each iteration, where removal is due to a node already having an output (whether the vertex is in the MIS or not). As in many uses of this method, we need to reduce the number of random choices that are made in order to be able to efficiently \emph{compute} the above assignment.

However, we need to overcome several obstacles. First, we need to reduce the search space of a good assignment to the random choices of the nodes, by showing that pairwise independence (see, e.g.,~\cite[Chapter 13]{UpfalMBook}) is sufficient for the algorithm to work. Unfortunately, this does not hold directly in the original algorithm.

The first key ingredient is a slight modification of the constants used by Ghaffari's algorithm. Ghaffari's analysis is based on a definition of \emph{golden nodes}, which are nodes that have a constant probability of being removed in the given phase. We show that this removal-probability guarantee holds also with pairwise independence upon our slight adaptation of the constants used by the algorithm.

Second, the shattering effect that occurs after $O(\log \Delta)$ rounds of Ghaffari's algorithm with \emph{full independence}, no longer holds under pairwise independence. Instead, we take advantage of the fact that in the congested clique model, once the remaining graph has a linear number of edges then the problem can be solved locally in constant many rounds using Lenzen's routing algorithm~\cite{Lenzen13}. Thus, we modify the algorithm so that after $O(\log \Delta)$ rounds, the remaining graph (containing all undecided nodes) contains $O(n)$ edges. The crux in obtaining this is that during the first $O(\log \Delta)$ phases, we favor the removal of \emph{old} nodes, which, roughly speaking, are nodes that had many rounds in which they had a good probability of being removed. This prioritized (or biased) removal strategy allows us to employ an amortized (or accounting) argument to claim that every node that survives $O(\log \Delta)$ rounds, can blame a distinct set of $\Delta$ nodes for not being removed earlier. Hence, the total number of remaining nodes is bounded by $O(n/\Delta)$, implying a remaining number of edges of $O(n)$.

To simulate the $O(\log \Delta)$ randomized rounds of Ghaffari's algorithm, we enjoy the small search space (due to pairwise independence) and employ the method of conditional expectations on a random seed of length $O(\log n)$.
Note that once we start conditioning on random variables in the seed, the random choices are no longer pairwise independent as they are in the unconditioned setting. However, we do not use the pairwise independence in the conditioning process. That is, the pairwise independence is important in showing that the \emph{unconditional expectation} is large, and from that point on the conditioning does not reduce this value.

As typical in MIS algorithms, the probability of a node being removed stems from the random choices made in its $2$-neighborhood. With a logarithmic bandwidth, collecting this information is too costly. Instead, we use a pessimistic estimator to \emph{bound} the conditional probabilities rather than compute them.

Finally, to make the decision of the partial assignment and inform the nodes, we leverage the power of the congested clique by having a leader that collects the relevant information for coordinating the decision regarding the partial assignment. In fact, the algorithm works in the more restricted \emph{Broadcast Congested Clique} model, in which a node must send the same $O(\log n)$-bit message to all other nodes in any single round.
Carefully placing all the pieces of the toolbox we develop, gives the following.
\newcommand{\TheoremMIS}
{
There is a deterministic MIS algorithm for the broadcast congested clique model that completes in $O(\log \Delta \log n)$ rounds.
}
\begin{theorem}
\label{theorem:MIS}
\TheoremMIS
\end{theorem}

If the maximal degree satisfies $\Delta=O(n^{1/3})$ then we can improve the running time in the congested clique model. In fact, under this assumption we can obtain the same complexity for $(\Delta+1)$-coloring.
\newcommand{\TheoremMISBoundedDelta}
{
If $\Delta=O(n^{1/3})$ then there is a deterministic MIS algorithm (and a $(\Delta+1)$-coloring algorithm) for the congested clique model that completes in $O(\log \Delta)$ rounds.
}
\begin{theorem}
\label{theorem:MISboundedDelta}
\TheoremMISBoundedDelta
\end{theorem}

Our techniques immediately extend to the CONGEST model. In that setting we show that MIS can be computed in $O(D \cdot \log^2 n)$ rounds where $D$ is the diameter of the graph. Here, we simulate $O(\log n)$ rounds of Ghaffari's algorithm rather than $O(\log \Delta)$ rounds as before. Each such randomized round is simulated by using $O(D\cdot \log n)$ deterministic rounds in which the nodes compute a $O(\log n)$ seed. Computing each bit of the seed, requires aggregation of the statistics to a leader which can be done in $O(D)$ rounds, and since the seed is of length $O(\log n)$, we have the following:
\newcommand{\TheoremMISCongest}
{
There is a deterministic MIS algorithm for the CONGEST model that completes in $O(D\log^2 n)$ rounds.
}
\begin{theorem}
\label{theorem:MIScongest}
\TheoremMISCongest
\end{theorem}

The significance of the latter is that it is the first deterministic MIS algorithm in CONGEST to have only a polylogarithmic gap compared to its randomized counterpart when $D$ is polylogarithmic. Notice that this logarithmic complexity is the best that is known even in the LOCAL model. In \cite{PanconesiS96} it is shown that an MIS can be computed deterministically in $O(2^{\sqrt{\log n}})$ rounds via network decomposition, which is super-polylogarithmic in $n$. Moreover, the algorithm requires large messages and hence is suitable for the LOCAL model but not for CONGEST.
Focusing on deterministic algorithms in CONGEST, the only known non-trivial solution is to use any $(\Delta+1)$-coloring algorithm running in $O(\Delta + \log^* n)$ rounds (for example \cite{barenboim2009distributed,Barenboim15}) to obtain the same complexity for deterministic MIS in CONGEST (notice that there are faster coloring algorithms, but the reduction has to pay for the number of colors anyhow). Our $O(D\log^2 n)$-round MIS algorithm is therefore unique in its parameters.

\paragraph{Multiplicative Spanners:} We further exemplify our techniques in order to derandomize the Baswana-Sen algorithm for constructing a $(2k-1)$-spanner. Their algorithm runs in $O(k^2)$ rounds and produces a $(2k-1)$-spanner with $O(kn^{1+1/k})$ edges. In a nutshell, the algorithm starts with a clustering defined by all singletons and proceeds with $k$ iterations, in each of which the clusters get sampled with probability $n^{1/k}$ and each node joins a neighboring sampled cluster or adds edges to unsampled clusters.

We need to make several technical modifications of our tools for this to work. The key technical difficulty is that we cannot have a single target function. This arises from the very nature of spanners, in that a small-stretch spanner always exists, but the delicate part is to balance between the stretch and the number of edges. This means that a single function which takes care of having a good stretch alone will simply result in taking all the edges into the spanner, as this gives the smallest stretch. We overcome this challenge by defining two types of bad events which the algorithm tries to avoid simultaneously. One is that too many clusters get sampled, and the other is that too many nodes add too many edges to the spanner in this iteration. The careful balance between the two promises that we can indeed get the desired stretch and almost the same bound on the number of edges.

Additional changes we handle are that when we reduce the independence, we cannot go all the way down to pairwise independence and we need settle for $d$-wise independence, where $d=\Theta(\log n)$. Further, we can improve the iterative procedure to handle chunks of $\log n$ random bits, and evaluate them in parallel by assigning a different leader to each possible assignment for them. A careful analysis gives a logarithmic overhead compared to the original Baswana-Sen algorithm, but we also save a factor of $k$ since the congested clique allows us to save the $k$ rounds needed in an iteration of Baswana-Sen for communicating with the center of the cluster. This gives the following.

\newcommand{\SpannerTheorem}
{
There is a deterministic $(2k-1)$-spanner algorithm for the congested clique model that completes in $O(k\log n)$ rounds and produces a spanner with $O(kn^{1+1/k}\log n)$ edges.
}
\begin{theorem}
\label{theorem:Spanner}
\SpannerTheorem
\end{theorem}

As in the MIS algorithm, the above algorithm works also in the broadcast congested clique model, albeit here we lose the ability to parallelize over many leaders and thus we pay another logarithmic factor in the number of rounds, resulting in $O(k\log^2 n)$ rounds.

\subsection{Related Work}
\paragraph{Distributed computation of MIS.} The complexity of finding a maximal independent set is a central problem in distributed computing and hence has been extensively studied. The $O(\log n)$-round randomized algorithms date back to 1986, and were given by Luby~\cite{luby1986simple}, Alon et al.~\cite{alon1986fast} and Israeli and Itai~\cite{IsraelI86}. A recent breakthrough by Ghaffari~\cite{ghaffari2016improved} obtained a randomized algorithm in $O(\log\Delta)+2^{O(\sqrt{\log \log n})}$ rounds. The best deterministic algorithm is by Panconesi and Srinivasan~\cite{panconesi1992improved}, and completes in $2^{O(\sqrt{\log n})}$ rounds.
On the lower bound side, Linial~\cite{linial1992locality} gave an $\Omega(\log^* n)$ lower bounds for $3$-coloring the ring, which also applies to finding an MIS. Kuhn et al.~\cite{KuhnMW16} gave lower bounds of $\Omega(\sqrt{\log n/\log\log n})$ and $\Omega(\sqrt{\log \Delta/\log\log \Delta})$ for finding an MIS.

Barenboim and Elkin~\cite{barenboim2013monograph} provide a thorough tour on coloring algorithms (though there are some additional recent results). An excellent survey on local problems is given by Suomela~\cite{Suomela13}.

\paragraph{Distributed constructions of spanners.}
The construction of spanners in the distribute setting has been studied extensively
both in the randomized and deterministic setting \cite{derbel2006fast,derbel2007deterministic,DerbelGPV08, derbel2009local,pettie2010distributed}.
A randomized construction was given by Baswana and Sen in \cite{baswana2007simple}. They show that their well-known centralized algorithm can be implemented in the distributed setting even when using small message size. In particular, they show that a $(2k-1)$ spanner with an expected number of $O(n^{1+1/k})$ edges can be constructed in $O(k^2)$ rounds in the CONGEST model. Derandomization of similar randomized algorithms has been addressed mainly in the \emph{centralized} setting \cite{RodittyTZ05}. We emphasize that we need entirely different techniques to derandomize the Baswana-Sen algorithm compared with the centralized derandomization of  \cite{RodittyTZ05}.

The existing \emph{deterministic} distributed algorithms for spanner are not based on derandomization of the randomized construction. They mostly use messages of \emph{unbounded} size and are mainly based on sparse partitions or network decomposition. The state of the art is due to Derbel et al~\cite{DerbelGPV08}.
They provide a \emph{tight} algorithm for constructing $(2k-1)$-spanners with optimal stretch, size and construction time of $k$ rounds. This was complemented by a matching lower bound, showing that any (even randomized) distributed algorithm requires $k$ rounds in expectation.  Much less efficient deterministic algorithms are known for the CONGEST model. The best bounds for constructing a $(2k-1)$-spanner are $O(n^{1-1/k})$ rounds due to \cite{derbel2010sublinear}.

\paragraph{Algorithms in the congested clique.}
The congested clique model was first addressed in Lotker et al.~\cite{lotker2003mst}, who raised the question of whether the global problem of constructing a minimum spanning tree (MST) can be solved faster on a communication graph with diameter $1$. Since then, the model gained much attention, with results about its computational power given by Drucker et al.~\cite{DruckerKO13}, faster MST algorithms by Hegeman et al.~\cite{HegemanPPSS15} and Ghaffari and Parter~\cite{Ghaffari-Parter-MST}, distance computation by Nanongkai~\cite{Nanongkai14,HenzingerKN16} and Holzer and Pinsker~\cite{HolzerP15}, subgraph detection by Dolev et al.~\cite{DolevLP12}, algebraic computations by Censor-Hillel et al.~\cite{Censor-HillelKK15}, and routing and load balancing by Lenzen~\cite{Lenzen13}, Lenzen and Wattenhoffer~\cite{LenzenW11}, and Patt-Shamir and Teplitsky~\cite{Patt-ShamirT11}. Local problems were addressed by Hegeman et al~\cite{HegemanPS14} who study ruling sets. Connections to the MapReduce model is given by Hegeman and Pemmaraju~\cite{HegemanP14}.

\paragraph{Derandomization in the parallel setting.}
Derandomization of local algorithms has attracted much attention in the \emph{parallel} setting \cite{alon1986fast,IsraelI86,motwani1989probabilistic,pantziou1989fast,KarpWigderson, goldberg1989new,berger1991simulating, han1996fast,kliemann2008parallel,chandrasekaran2013deterministic}. Luby~\cite{luby1993removing}  showed that his MIS algorithm (and more) can be derandomized in the PRAM model using $O(m)$ machines and $O(\log^3 n \log \log n)$ time. In fact, this much simpler algorithm can also be executed on the congested clique model, resulting in an $O(\log^4 n)$ running time.
Similar variants of derandomization for MIS, maximal matching and $(\Delta+1)$-coloring were presented in \cite{alon1986fast, IsraelI86}.  Berger and Rompel \cite{berger1991simulating} developed a general framework for removing randomness from RNC algorithms when polylogarithmic independence is sufficient.
The parallel setting bears some similarity to the all-to-all communication model but
the barriers in these two models are different mainly because the complexity measure in the parallel setting is the computation time while in our setting local computation is for free. This raises the possibility of obtaining much better results in the congested clique model compared to what is known in the parallel setting.

\section{Preliminaries and Notation}
\label{section:preliminaries}
Our derandomization approach consists of first reducing the independence between the coin flips of the nodes. Then, we find some target function we wish to maintain during each iteration of the derandomized algorithm. Finally, we find a pessimistic estimator for the target function and apply the method of conditional expectations to get a deterministic algorithm. Below we elaborate upon the above ingredients.

\paragraph{$d$-wise independent random variables.}
In the algorithms we derandomize in the paper, a node $v \in V$ flips coins with probability $p$ of being heads. As we show, it is enough to assume only $d$-wise independence between the coin flips of nodes. We show how to use a randomness \emph{seed} of only $t=d \max \set{\log n, \log 1/p}$ bits to generate a coin flip for each $v\in V$, such that the coin flips are $d$-wise independent.

We first need the notion of $d$-wise independent hash functions as presented in \cite{Vadhan12}.
\begin{definition}[{\cite[Definition 3.31]{Vadhan12}}]
	\label{def: d-wise independent}
	For	$N,M,d \in \mathbb{N} $ such that $d \leq N$, a family of functions $\Hcal = \set{h : [N] \rightarrow
		[M]}$ is $d$-wise independent if for all distinct $x_1,x_2,...,x_d \in [N],$ the
	random variables $H(x_1),...,H(x_d)$ are independent and uniformly distributed
	in $[M]$ when $H$ is chosen randomly from $\Hcal$.
\end{definition}

In \cite{Vadhan12} an explicit construction of $\Hcal$ is presented, with parameters as stated in the next Lemma.

\begin{lemma}[{\cite[Corollary 3.34]{Vadhan12}}]
	\label{lem: d-wise independent}
	For every $\gamma,\beta,d \in \mathbb{N},$ there is a family of $d$-wise independent functions $\mathcal{H}_{\gamma,\beta} = \set{h : \set{0,1}^\gamma \rightarrow \set{0,1}^\beta}$ such that choosing a random function from $\mathcal{H}_{\gamma,\beta}$ takes $d \cdot \max \set{\gamma,\beta}$ random bits, and evaluating
	a function from $\mathcal{H}_{\gamma,\beta}$ takes time $poly(\gamma,\beta,d)$.
\end{lemma}

Let us now consider some node $v \in V$ which needs to flip a coin with probability $p$ that is $d$-wise independent with respect to the coin flips of other nodes. Using Lemma~\ref{lem: d-wise independent} with parameters $\gamma=\ceil{\log n}$ and $\beta=\ceil{\log {1/p}}$, we can construct $\Hcal$ such that every function $h\in \Hcal$ maps the ID of a node to the result of its coin flip. Using only $t$ random bits we can flip $d$-wise independent biased coins with probability $p$ for all nodes in $v$.

We define $\ybar$ to be a vector of $t$ random coins. Note we can also look at $\ybar$ as a vector of length $t / \log n$ where each entry takes values in $[\log n]$. We use the latter when dealing with $\ybar$. From $\ybar$ each node $v$ can generate its random coin toss by accessing the corresponding $h \in \Hcal$ and checking whether $h(ID(v)) = 0$. From Definition~\ref{def: d-wise independent} it holds that $Pr[h(ID(v)) = 0] = 1/p$, as needed.

\paragraph{The method of conditional expectations.}
Next, we consider the method of conditional expectations. Let $\phi: A^\ell \rightarrow \mathbb{R}$, and let $X = (X_1,...,X_\ell)$ be a vector of random variables taking values in $A$. If $E[\phi(X)] \geq \alpha$ then there is an assignment of values $Z=(z_1,..., z_\ell)$ such that $\phi(Z) \geq \alpha$. We describe how to \emph{find} the vector $Z$.
We first note that from the law of total expectation it holds that
$E[\phi(X)] = \sum_{a\in A} E[\phi(X) \mid X_1 = a]Pr[X_1=a]$, and therefore for at least some $a \in A$ it holds that $E[\phi(X) \mid X_1 = a] \geq \alpha$. We set this value to be $z_1$. We then repeat this process for the rest of the values in $X$, which results in the vector $Z$. In order for this method to work we need it to be possible to \emph{compute} the conditional expectation of $\phi(X)$.
We now wish to use the method of conditional expectations after reducing the number of random bits used by the algorithm. Let us denote by $\rhobar$ the original vector of random bits used by the algorithm. Taking $\ybar$ as before to be the seed vector for $\rhobar$, we have that $\rhobar$ is a function of $\ybar$.
We need to be able to compute $E[\phi(\rhobar(\ybar)) \mid y[1]=a_1, \dots, y[i]=a_i]$ for all possible values of $i$ and $a_j, j\leq i$.

Computing the conditional expectations for $\phi$ might be expensive. For this reason we use a \emph{pessimistic estimator}. A pessimistic estimator of $\phi$ is a function $\phi': A^\ell \rightarrow \mathbb{R}$ such that that for all values of $i$ and $a_j,j\leq i$ it holds that $E[\phi(\rhobar(\ybar)) \mid y_1=b_1, \dots, y_i=b_i] \geq E[\phi'(\rhobar(\ybar)) \mid y_1=b_1, \dots, y_i=b_i]$. If $\phi'$ is a pessimistic estimator of $\phi$, then we can use the method of conditional expectations on $\phi'$ and obtain $z_1,\dots,z_n$, such that $\phi(z_1,\dots, z_n) \geq \phi'(z_1,\dots, z_n) \geq \alpha$.


\paragraph{Lenzen's routing algorithm.}
One important building block for algorithms for the congested clique model is Lenzen's routing algorithm~\cite{Lenzen13}. This procedure guarantees that if there is a component of an algorithm in which each node needs to send at most $O(n\log n)$ bits and receive at most $O(n \log n)$ bits then $O(1)$ rounds are sufficient. This corresponds to sending and receiving $O(n)$ pieces of information of size $O(\log n)$ for each node. Intuitively, this is easy when each piece of information of a node has a distinct destination, by a direct message. The power of Lenzen's routing is that the source-destination partition does not have to be uniform, and that, in fact, it can be not predetermined.

\section{Deterministic MIS}
\label{section:MIS}
We consider the following modification of the randomized algorithm of Ghaffari \cite{ghaffari2016improved}, where the constants are slightly changed. The algorithm of Ghaffari consists of two parts. The first part (shown to have a good local complexity) consists of $O(\log \Delta)$ phases, each with $O(1)$ rounds. 
After this first part, it is shown that sufficiently many nodes are removed from the graph. The MIS for what remains is computed in the second part deterministically in time $O(2^{\sqrt{\log\log n}})$.\\ 

\begin{mdframed}[hidealllines=false,backgroundcolor=gray!25]
\textbf{The modification to the first part of Ghaffari's MIS Algorithm}\\
Set $p_0(v)=1/4$.
$$
p_{t+1}(v)=
\begin{cases}
1/2\cdot p_t(v), \mbox{~~~if~} d_t(v)\geq 1/2\\
\min\{2p_t(v),1/4\}, \mbox{~~~if~} d_t(v)< 1/2,
\end{cases}
$$
where $d_t(v)=\sum_{u \in N(v)}p_t(u)$ is the \emph{effective degree} of node $v$ in phase $t$.  In each phase $t$, the node $v$ gets marked with probability $p_t(v)$ and if none of its neighbors is marked, $v$ joins the MIS and gets removed along with its neighbors.
\end{mdframed}

\subsection{An $O(\log \Delta)$ round randomized MIS algorithm in the congested clique}
\label{subsec:MISlogDelta}
We begin by observing that in the congested clique, what remains after $O(\log \Delta)$ phases of Ghaffari's algorithm can be solved in $O(1)$ rounds. This provides an improved \emph{randomized} runtime compared to~\cite{ghaffari2016improved}, and specifically, has no dependence on $n$.

The algorithm consists of two parts. In the first part, we run Ghaffari's algorithm for $O(\log \Delta)$ phases. We emphasize that this works with both Ghaffari's algorithm and with our modified Ghaffari's algorithm, since the values of the constants do not affect the running time and correctness of the randomized first part of the algorithm.

Then, in the second part, a leader collects all surviving edges and solves the remaining MIS deterministically on that subgraph.
We show that the total number of edges incident to these nodes is $O(n)$ w.h.p., and hence using the deterministic routing algorithm of Lenzen~\cite{Lenzen13}, the second part can be completed in $O(1)$ rounds w.h.p.

\begin{lemma}[{~\cite[Theorem 3.1]{ghaffari2016improved}}]
\label{lem:ghaffariremove}
For every $\epsilon \in (0,1)$, there exists a constant $c'$, such that for each node $v$, the probability that $v$ has not made its decision within the first $c' \cdot (\log \Delta+\log 1/\epsilon)$ phases is at most $\epsilon$.
\end{lemma}
Since the decision whether to join the MIS or to be removed by having a neighbor in the MIS depends only on the $2$-neighborhood of a node, decisions made by nodes that are at least $4$ hops from each other are independent. We make use of the following variant of Chernoff's bound.

\begin{fact}[A Chernoff bound with bounded dependency~\cite{Pemmaraju01}]
\label{fc:chernoffbounded}
Let $X_1, \ldots, X_n$ denote a set of binary random variables with bounded dependency $\widehat{d}$, and let $\mu=\mathbb{E}(\sum_i X_i)$. Then:
$$\Pr[\sum_i X_i \geq (1+\delta)\mu]\leq O(\widehat{d}) \cdot e^{-\Theta(\delta^2 \mu/\widehat{d})}.$$ 
\end{fact}

\begin{corollary}
\label{cor:shater}
After $O(\log \Delta)$ phases, the number of edges incident to nodes that have not made their decisions is $O(n)$.
\end{corollary}

\begin{proof}
By Lemma~\ref{lem:ghaffariremove}, there is a constant $c$, such that the probability that a node has not made its decision after $O(\log \Delta)$ phases is at most $1/\Delta^c$. Letting $X$ denote the  random variable of the number of nodes that have not made their decision after $O(\log \Delta)$ phases gives that $\mu=\E[X]=n/\Delta^{c}$.

Each node is mutually independent from all nodes except in its $4$-neighborhood, in which there are up to $\Delta^4$ nodes. By using Fact~\ref{fc:chernoffbounded} with $\delta=\Theta(\Delta^{c-1})$ and $\widehat{d}=\Delta^{4}$, and taking $c$ to be large enough, we get that
$$\Pr[X \geq n/\Delta]\leq O(\Delta^{4}) \cdot e^{-\Theta(n)}.$$
Hence, w.h.p., there are $O(n/\Delta)$ remaining nodes and therefore the number of their incident edges is at most $O(n)$.
\end{proof}

We conclude:
\begin{theorem-repeat}{theorem:randcongestedclique}
\TheoremMISlogDelta
\end{theorem-repeat}

Note that the proof of Corollary~\ref{cor:shater} cannot be extended to the case of pairwise independence, which is needed for derandomization, since the concentration guarantees are rather weak. For this, we need to develop in the following section new machinery.

\subsection{Derandomizing the modified MIS algorithm}
We first turn to consider the modified Ghaffari's algorithm when the random decisions made by the nodes are only \emph{pairwise} independent.

\subsubsection{Ghaffari's algorithm with pairwise independence}
We review the main terminology and notation from~\cite{ghaffari2016improved}, up to our modification of constants.
Recall that $d_t(v)=\sum_{u \in N(v)}p_t(u)$. A node $v$ is called \emph{light} if $d_t(v)<1/4$.

\paragraph{Golden phases and golden nodes.}
We define two types of \emph{golden phases} for a node $v$. This definition is a modification of the corresponding definitions in~\cite{ghaffari2016improved}.

\begin{definition}
\label{dec:golden}
\item{Type-1 golden phase:} $p_t(v)=1/4$ and $d_t(v)\leq 1/2$;
\item{Type-2 golden phase:} $d_t(v)>1/4$ and at least $d_t(v)/10$ of it arises from light nodes.
\end{definition}

A node $v$ is called \emph{golden} in phase $t$, if phase $t$ is a golden phase for $v$ (of either type).
Intuitively, a node $v$ that is golden in phase $t$ is shown to have a constant probability of being removed.
Specifically, in a golden phase of type-1, $v$ has a constant probability to join the MIS and in a golden phase of type-2, there is a constant probability that $v$ has a neighbor that joins the MIS and hence $v$ is removed.

We now prove the analogue of Lemma 3.3 in \cite{ghaffari2016improved} for the setting in which the coin flips made by the nodes are not completely independent but are only \emph{pairwise independent}. We show that a golden node for phase $t$ is still removed with constant probability even under this weaker bounded independence guarantee.

\begin{lemma}[Type-1 golden nodes with pairwise independence]
\label{lem:gonepairwise}
Consider the modified Ghaffari's algorithm with pairwise independent coin flips 
If $t$ is a type-1 golden phase for a node $v$, then $v$ joins the MIS in phase $t$ with probability at least $1/8$. 
\end{lemma}

\begin{proof}
In each type-1 golden phase, $v$ gets marked with probability $p_t(v)=1/4$. By the inclusion-exclusion principle, the probability that $v$ is marked but none of its neighbors is marked in phase $t$ can be bounded by:
\begin{eqnarray*}
\Pr[v \mbox{~is the only marked node in~} N(v)] &\geq& p_t(v)-\sum_{u \in N(v)}p_t(v)\cdot p_t(u)\\
&\geq&
p_t(v)(1-d_t(v))\geq 1/4(1-1/2)=1/8.~
\end{eqnarray*}
Hence, $v$ joins the MIS with probability at least $1/8$.
\end{proof}

We next show that also the golden nodes of type-2 are removed with constant probability assuming only pairwise independence.
\begin{lemma}[Type-2 golden nodes with pairwise independence]
\label{lem:gtwopairwise}
Consider the modified Ghaffari's algorithm with pairwise independent coin flips. If $t$ is a type-2 golden phase for a node $v$ then $v$ is removed in phase $t$ with probability at least $\alpha=1/160$.
\end{lemma}

\begin{proof}
For node $v$, fix a subset of light neighbors $W(v)\subseteq N(v)$ satisfying that $\sum_{u \in W(v)}p_t(u) \in [1/40, 1/4]$. Such a set exists because by the definition of a type-2 golden phase, the sum of probabilities of light neighbors of $v$ is at least $1/40$, and every single probability is at most $1/4$ by the constants taken in the algorithm (this probability either halves in each phase, or is bounded from above by $1/4$).

For $u \in W(v)$, let $\Upsilon_t(v,u)$ denote the indicator variable of the event that in phase $t$ the following happens: $u$ gets marked, \emph{none} of $u$'s neighbors get marked and $u$ is the \emph{only} neighbor of $v$ in $W(v)$ that got marked. For node $u,v$ and phase $t$, let $m_{u,v,t}$ be the indicator random variable that both $u$ and $v$ get marked in phase $t$
Due to the pairwise independence, we have:
$\Pr[m_{u,v,t}=1]=p_t(u) \cdot p_{t}(v)$. Hence, the probability of the event indicated by $\Upsilon_t(v,u)$ can be bounded by:
\begin{eqnarray*}
\Pr[\Upsilon_t(v,u)=1]& \geq & p_{t}(u)-\sum_{w \in N(u)}\Pr[m_{u,w,t}=1]-\sum_{u' \in W\setminus\{u\}}\Pr[m_{u,u',t}=1]
\\&=&
p_{t}(u)-\sum_{w \in N(u)}p_t(u)p_t(w)-\sum_{u' \in W\setminus\{u\}}p_t(u)p_t(u')
\\ &\geq&
p_{t}(u)\left( 1- \sum_{w \in N(u)}p_t(w) -\sum_{u' \in W\setminus\{u\}}p_t(u') \right)
\\ &\geq&
p_{t}(u)\left( 1- d_t(u)-1/4 \right)\geq p_t(u)\left( 1-1/2-1/4 \right)= 1/4 \cdot p_t(u).
\end{eqnarray*}
Since the events indicated by $\Upsilon_t(v,u)$ are mutually exclusive for different $u,u' \in N(v)$, it holds that
the probability that $v$ gets removed is at least
\begin{eqnarray*}
\Pr[v\mbox{~is removed in phase~} t]\geq \sum_{u \in W(v)}\Pr[\Upsilon_t(v,u)=1]\geq 1/4 \cdot \sum_{u \in W(v)} p_t(u)\geq 1/160~.
\end{eqnarray*}
\end{proof}

Finally, we claim the analogue of Lemma 3.2 in~\cite{ghaffari2016improved}.
\begin{lemma}
\label{lem:manygold}
Consider the modified Ghaffari's algorithm with pairwise independent randomness and $\epsilon=1/\Delta^c$. For a large enough $c'$, for every $v$, at the end of $c' \cdot \log \Delta$ phases, either $v$ has joined the MIS, or it has a neighbor in the MIS, or at least one of its golden phase counts reached $c \cdot \log \Delta$. 
\end{lemma}

The proof here is exactly that same as in \cite{ghaffari2016improved}. The reason is that the proof does not assume independence and is only affected by the update rule of the probabilities. Note that similarly to \cite{ghaffari2016improved}, we had to define type-2 with a threshold on $d_t(v)$ which is factor $2$ smaller than that of type-1.
As a result, the following holds in the pairwise independence setting:
\begin{lemma}
Within $O(\log \Delta)$ phases, every node remains with probability at most $1/\Delta$.
\end{lemma}
Recall that the proof of Lemma~\ref{cor:shater} cannot be extended to pairwise independence since the concentration guarantees are rather weak. Our algorithm will use pairwise independence but with some crucial modifications required in order to guarantee that after $O(\log \Delta)$ phases, only $O(n/\Delta)$ nodes remain undecided.
%

\subsubsection{Deterministic $O(\log \Delta \log n)$-Round Algorithm in the Congested Clique}
We now turn to consider the derandomization procedure.
We show the following:
\begin{theorem-repeat}{theorem:MIS}
\TheoremMIS
\end{theorem-repeat}

\paragraph{The challenge.}
Consider phase $t$ in the modified Ghaffari's algorithm and let $V_t$ 
be the set of golden nodes in this phase. Our goal is to select additional nodes into the MIS so that at least a constant fraction of the golden nodes are removed.
Let $v_1, \ldots, v_{n'}$ be the nodes that are \emph{not} removed in phase $t$.
For each node, define corresponding random variables $x_1, \ldots, x_{n'}$ indicating whether $v_i$ is marked in phase $t$ or not. Let $X_i=(x_1=b_1, \ldots, x_{i}=b_{i})$ define a partial assignment for the nodes $v_1, \ldots, v_{i}$ (i.e., whether or not they are in the MIS in phase $t$). Let $X_0=\emptyset$ denote the case where none of the decisions is fixed.

For a golden node $v$, let $r_{v,t}$ be the random variable indicating whether $v$ gets removed in phase $t$, and let $R_t$ be the random variable of the number of removed golden nodes.
By linearity of expectation, $\mathbb{E}(R_t)=\sum_{v} \mathbb{E}(r_{v,t})$ is the \emph{expected} number of removed golden nodes in phase $t$.
By Lemmas~\ref{lem:gonepairwise} and~\ref{lem:gtwopairwise}, there is a constant $c$ such that
$\mathbb{E}(R_t)\geq c \cdot |V_{t}|$.

Potentially, we could aim for the following: Given the partial assignment $X_i$, compute the two expectations of the number of removed golden nodes conditioned on the two possible assignments for $x_{i+1}$, $\mathbb{E}(R_t \mid~ X_i, x_{i+1}=0)$ and $\mathbb{E}(R_t \mid~ X_i, x_{i+1}=1)$, and choose $x_{i+1}$ according to the larger expectation.  

However, towards the above goal we face the following main challenges.
\begin{enumerate}
\item[(C1)]
The value of $R_t$ cannot be easily computed, since when using probabilities of neighboring nodes we might be double-counting: a node might be removed while having more than single neighbor that joins the MIS.
\item[(C2)]
The search space of size $2^n$ is too large and in particular, the conditional expectation computation consists of $n$ steps.
\item[(C3)]
Even when using pairwise independence to enjoy a $O(\log n)$-bit seed, searching a good point in a space of size $O(\poly~n)$ in a brute force manner cannot be done efficiently in the congested clique.
\item[(C4)]
Despite our proof that golden nodes are removed with constant probability even with pairwise independence, it is still not clear how to implement the second part of the MIS algorithm, because showing that only $O(n/\Delta)$ nodes survive cannot be done with pairwise independence. That is, the proof of Corollary~\ref{cor:shater} inherently needs full independence.
\end{enumerate}

Addressing (C4) requires the use of a priority-based scheme for choosing the nodes that join the MIS, which requires a novel age-based weighting approach to be added to the MIS algorithm. Next, we describe our main derandomization tools and then provide our algorithm.

\paragraph{Derandomization tools.}
We define a pessimistic estimator to the conditional expectation $\mathbb{E}(R_t ~\mid~ X_i)$, which can be computed efficiently in our model. Then, we describe how to reduce the search space using pairwise independence. In our algorithm, the nodes will apply the method of conditional expectations on the estimator in order to find a ``good" seed of length $O(\log n)$.

\paragraph{Tool 1: The pessimistic estimator function.}
Consider phase $t$ and recall that $V_t$ are the golden nodes in this phase.
Similarly to the clever approach of~\cite{luby1986simple,luby1993removing},
we define a variable $\psi_{v,t}$ that will satisfy that $r_{v,t} \geq \psi_{v,t}$. 
The idea is to account for a removed node of type-2 only in the case that it is removed because a \emph{single} one of its neighbors joins the MIS. Since this can only occur for one of its neighbors, we avoid double-counting when computing the probabilities. This allows us to cope with challenge (C1).

Let $m_{v,t}$ be the random variable indicating the event that $v$ is \emph{marked}.
Let $m_{v,u,t}$ indicate the event that both $u$ and $v$ are marked.
Define
$$
\psi_{v,t}=
\begin{cases}
m_{v,t}-\sum_{u \in N(v)}m_{v,u,t}, \mbox{~if~} v \mbox{~is of type-1}.\\ \\
\sum_{u \in N(v)}\left(m_{v,u,t}-\sum_{w \in N(u)}m_{u,w,t}-\sum_{w' \in N(v)\setminus \{u\}}m_{u,w',t}\right), \mbox{~if~} v \mbox{~is of type-2}.
\end{cases}
$$
Denoting $\Psi_t=\sum_{v \in V_t}\psi_{v,t}$ gives that $\Psi_t$ is a lower bound on the number of removed golden nodes, i.e., $\Psi_t \leq R_t$.
For a partial assignment $X_i=(x_1=b_1, \ldots, x_{i}=b_i)$ indicating which of the nodes are in the MIS, we have
\begin{equation}
\label{eq:psi-golden}
\mathbb{E}(\psi_{v,t} ~\mid~ X_i)=
\begin{cases}
\Pr[m_{v,t}=1 ~\mid ~X_i]-\sum_{u \in N(v)}\Pr[m_{v,u,t}~\mid ~X_i], \mbox{~~if~} v \mbox{~is of type-1}.\\ \\
\sum_{u \in N(v)}(\Pr[m_{v,u,t}=1 ~\mid~ X_i]-\\
\sum_{w \in N(u)}\Pr[m_{u,w,t}=1 ~\mid~ X_i]-\\
\sum_{w' \in W(v)\setminus \{u\}}\Pr[m_{u,w',t}=1 ~\mid~ X_i)], \mbox{~if~} v \mbox{~is of type-2},
\end{cases}
\end{equation}
where $W(v) \subseteq N(v)$ is a subset of $v$'s neighbors satisfying that
$\sum_{w \in W(v)} p_t(w) \in [1/40,1/4]$ (as used in the proof of  Lemma~\ref{lem:gtwopairwise}).
By Lemmas~\ref{lem:gonepairwise} and~\ref{lem:gtwopairwise}, it holds that $\mathbb{E}(\psi_{v,t})\geq \alpha$ for $v \in V_t$. Hence, we have that:
$$\mathbb{E}(r_{v,t})\geq \mathbb{E}(\psi_{v,t})\geq \alpha.$$
Since $r_{v,t}\geq \psi_{v,t}$ even upon conditioning on the partial assignment $X_i$, we get:
\begin{eqnarray*}
\mathbb{E}(R_{v,t} ~\mid~ X_i)&\geq &\mathbb{E}(\Psi_{t} ~\mid X_i) =\sum_{v \in V_t}{\E}(\psi_{v,t} ~\mid~ X_i)\geq \alpha \cdot |V_{t}|.
\end{eqnarray*}
Our algorithm will employ the method of conditional expectations on a \emph{weighted} version of $\mathbb{E}(\Psi_{t} ~\mid X_i)$, as will be discussed later.

\paragraph{Tool 2: Pairwise independence.}
We now combine the method of conditional expectations with a small search space.
We use Lemma \ref{lem: d-wise independent} with $d=2$, $\gamma=\Theta(\log n)$ and a prime number
$\beta=O(\log \Delta)$. This is because we need the marking probability, $p_t(v)$, to be $\Omega(1/\poly~\Delta)$.

Consider phase $t$. Using the explicit construction of Lemma~\ref{lem: d-wise independent}, if all nodes are given a shared random seed of length $\gamma$, they can sample a random hash function $h:\{0,1\}^\gamma \to \{0,1\}^\beta$ from $\mathcal{H}_{\gamma,\beta}$ which yields $n$ pairwise independent choices. Specifically, flipping a biased coin with probability of $p_t(v)$ can be trivially simulated using the hash value $h(ID_v)$ where $ID_v$ is a $O(\log n)$-bit ID of $v$.\footnote{Flipping a biased coin with probability $1/2^i$, is the same as getting a uniformly distributed number $y$ in $[1,b]$ and outputting $1$ if and only if $y \in [1, 2^{b-i}]$.} Since $h$ is a random function in the family, all random choices are pairwise independent and the analysis of of the golden phases goes through. This standard approach takes care of challenge (C2).

Even though using a seed of length $O(\log n)$ reduces the search space to be of polynomial size, still, exploring all possible $2^{O(\log n)}=O(n^c)$ seeds in a brute force manner is too time consuming. Instead, we employ the method of conditional expectations to find a \emph{good seed}. That is, we will consider $\mathbb{E}(\Psi_t ~\mid~ Y_i)$ where $Y_i=(y_1=b_1, \ldots, y_i=b_i)$ is a partial assignment to the \emph{seed} $Y=(y_1, \ldots, y_a)$. The crux here is that since a \emph{random} seed is good, then so is the expectation over seeds that are sampled uniformly at random. Hence, the method of conditional expectations will find a seed that is at least as good as the random selection. Specifically, we still use the pessimistic estimator of Equation (\ref{eq:psi-golden}), but we condition on the small seed $Y_i$ rather than on $X_i$. This addresses challenge (C3).

\paragraph{Tool 3: An age-based weighted adaptation.}
To handle challenge (C4), we compute the expectation of a weighted version of $\Psi_t$, which favors \emph{old} nodes where the age of a node is counted as the number of golden phases it experienced. Let $age(v)$ be the number of golden phases $v$ has till phase $t$ and recall that a golden node is removed with probability at least $\alpha$.
Define $\psi'_{v,t}= (1/(1-\alpha))^{age(v)} \cdot \psi_{v,t}$, and $\Psi'_{t} =\sum_{v \in V_t}\psi'_{v,t}$. We use the method of conditional expectations for:
\begin{equation}
\label{eq:weightedexp}
\mathbb{E}(\Psi'_{t} ~\mid Y_i) =\sum_{v \in V_t} \mathbb{E}(\psi'_{v,t} ~\mid~ Y_i)~,
\end{equation}
rather than for $\mathbb{E}(\Psi'_{t} ~\mid Y_i)$. The choice of this function will be made clear in the proof of Lemma~\ref{lem:amortized}.
%

\paragraph{Algorithm Description.}
The first part of the algorithm consists of $\Theta(\log \Delta)$ phases, where in phase $t$, we derandomize phase $t$ in the modified Ghaffari's algorithm using $O(\log n)$ deterministic rounds. In the second part, all nodes that remain undecided after the first part, send their edges to the leader using the deterministic routing algorithm of Lenzen.
The leader then solves locally and notifies the relevant nodes to join the MIS. In the analysis section, we show that after the first part, only $O(n/\Delta)$ nodes remain undecided, and hence the second part can be implemented in $O(1)$ rounds.

From now on we focus on the first part. Consider phase $t$ in the modified Ghaffari's algorithm. Note that at phase $t$, some of the nodes are already removed from the graph (either because they are part of the MIS or because they have a neighbor in the MIS). Hence, when we refer to nodes or neighboring nodes, we refer to the remaining graph induced on the undecided nodes.

Let $Y=(y_1, \ldots, y_\gamma)$ be the $\gamma$ random variables that are used to select a hash function and hence induce a deterministic algorithm. We now describe how to compute the value of $y_i$ in the seed, given that we already computed $y_1=b_1, \ldots, y_{i-1}=b_{i-1}$.  By exchanging IDs (of $\Theta(\log n)$ bits), as well as the values $p_t(v)$ and $d_t(v)$  with its neighbors, a node can check if it is a golden type-1 or type-2 node according to the conditions of Definition~\ref{dec:golden}. In addition, every node maintains a counter, $age(v)$ referred to as the age of $v$, which measures the number of golden phases it had so far.

Depending on whether the node $v$ is a golden type-1 or type-2 node,
based on Equation (\ref{eq:psi-golden}), it computes a lower bound on the conditional probability that it is removed given the partial seed assignment $Y_{i,b}=(y_1, \ldots, y_{i}=b)$ for every $b \in \{0,1\}$. These lower bound values are computed according to the proofs of Lemmas~\ref{lem:gonepairwise} and~\ref{lem:gtwopairwise}.

Specifically, a golden node $v$ of type-1, uses the IDs of its neighbors and their $p_t(u)$ values to compute the following:
$$\mathbb{E}(\psi_{v,t} ~\mid~ Y_{i,b})=\Pr[m_{v,t}=1 ~\mid~ Y_{i,b}]-\sum_{u \in N(v)}\Pr[m_{u,t}=1 ~\mid~ Y_{i,b}],$$
where $\Pr[m_{v,t}=1 ~\mid~ Y_{i,b}]$ is the conditional probability that $v$ is marked in phase $t$  (see Sec. \ref{sec:code} for full details about this computation).

For a golden node $v$ of type-2 the lower bound is computed differently. First, $v$ defines a subset of neighbors $W(v) \subseteq N(v)$, satisfying that $\sum_{w \in W(v)}p_t(w)\in [1/40,1/4]$, as in the proof of Lemma~\ref{lem:gtwopairwise}. Let $M_{t,b}(u)$ be the conditional probability on $Y_{i,b}$ that $u$ is marked but none of its neighbors are marked. Let $M_{t,b}(u,W(v))$ be the conditional probability on $Y_{i,b}$ that another node other than $u$ is marked in $W(v)$.\footnote{The term $M_{t,b}(u,W(v))$ is important as it is what prevents double counting, because the corresponding random variables defined by the neighbors of $v$ are mutually exclusive.} By exchanging the values $M_{t,b}(u)$, $v$ computes:
$$\mathbb{E}(\psi_{v,t} ~\mid~ Y_{i,b})=\sum_{u \in W(v)}\Pr[m_{u,t} =1 ~\mid~ Y_{i,b}]-M_{t,b}(u)-M_{t,b}(u,W(v)).$$

Finally, as in Equation (\ref{eq:weightedexp}), the node sends to the leader the values $\mathbb{E}(\psi'_{v,t} ~\mid~ Y_{i,b})=1/(1-\alpha)^{age(v)}\cdot \mathbb{E}(\psi_{v,t} ~\mid~ Y_{i,b})$ for $b \in \{0,1\}$.
The leader computes the sum of the $\mathbb{E}(\psi'_{v,t} ~\mid~ Y_{i,b})$ values of all golden nodes $V_t$, and declares that $y_{i}=0$ if $\sum_{v \in V_t}\mathbb{E}(\psi'_{v,t} ~\mid~ Y_{i,b})\geq \sum_{v \in V_t}\mathbb{E}(\psi'_{v,t} ~\mid~ Y_{i,b})$, and $y_{i}=1$ otherwise. This completes the description of computing the seed $Y$.

Once the nodes compute $Y$, they can simulate phase $t$ of the modified Ghaffari's algorithm. In particular, the seed $Y$ defines a hash function $h \in \mathcal{H}_{\gamma,\beta}$ and $h(ID(v))$ can be used to simulate the random choice with probability $p_t(v)$.  The nodes that got marked send a notification to neighbors and if none of their neighbors got marked as well, they join the MIS and notify their neighbors. Nodes that receive join notification from their neighbors are removed from the graph. This completes the description of the first part of the algorithm. For completeness, a pseudocode appears in Appendix~\ref{sec:MISdetailed}.

\paragraph{Analysis.}
The correctness proof of the algorithm uses a different argument than that of Ghaffari~\cite{ghaffari2016improved}. Our proof does not involve claiming that a constant fraction of the golden nodes are removed, because in order to be left with $O(n/\Delta)$ undecided nodes we have to favor removal of old nodes. The entire correctness is based upon the following lemma, which justifies the definition of the expectation given in Equation (\ref{eq:weightedexp}).

\begin{lemma}
\label{lem:amortized}
The number of undecided nodes after $\Theta(\log \Delta)$ phases is $O(n/\Delta)$ and hence the total number of edges incident to these nodes is $O(n)$.
\end{lemma}

\begin{proof}
Consider phase $t$, and let $V_t$ be the set of golden nodes in this phase. For a golden node $v \in V_t$, let $age(v)$ be the number of golden phases it had so far, excluding the current one. We refer to this quantity as the \emph{age} of the node. Hence, intuitively, a node is \emph{old} if it has experienced many golden phases.

By Lemma \ref{lem:manygold}, a node that remains undecided after $\beta \log \Delta$ phases (for a sufficiently large constant $\beta$) is of age at least $\Omega(\log \Delta)$. Our goal is to show that the number of old nodes (with $age(v)\geq \Omega(\log \Delta$)) is bounded by $n/\Delta$. At a high level, we show that each old node can \emph{blame} or \emph{charge} a distinct set of $\Delta$ nodes that are removed. Hence, there are total of at most $n/\Delta$ old nodes.

Recall that in phase $t$, given the golden nodes $V_t$, the original goal is to find an assignment to the nodes so that the number of removed nodes is at least as the expected number of the removed nodes, say, at least $\alpha \cdot |V_t|$. Our algorithm modifies this objective by \emph{favoring} the removal of \emph{older} nodes.
To do that, recall that $r_{t,v}$ is the random variable indicating the event that a golden node $v$ is removed at phase $t$, and $\psi_{v,t}$ is used to obtain a pessimistic estimator for $v$ being removed. The age of $v$ at phase $t$, $age(v)$, is the number of golden phases $v$ had till phase $t$.
The nodes compute the conditional expectation for:
$$\mathbb{E}(\Psi'_t)=\sum_{v \in V_t}\Pr[\psi_{v,t} \geq 1]\cdot (1/(1-\alpha))^{age(v)}.$$
By Lemma \ref{lem:gonepairwise} and \ref{lem:gtwopairwise}, $\Pr[\psi_{v,t}\geq 1]\geq \alpha$, hence:
$$\mathbb{E}(\Psi'_t)\geq \alpha \sum_{v \in V_t}(1/(1-\alpha))^{age(v)}.$$

For every phase $t$, let $sb_t=\sum_{v \in V_t}(1/(1-\alpha))^{age(v)}$ be the total weighted age of the nodes in $V_t$. Let $V_{t,0} \subseteq V_t$ be the set of unremoved nodes in $V_t$ and let $V_{t
,1} \subseteq V_t$ be the set of removed nodes in $V_t$ (i.e., the decided ones).
Let $sb_{t,0}=\sum_{v \in V_{t,0}}(1/(1-\alpha))^{age(v)}$ and $sb_{t,1}=\sum_{v \in V_{t,1}}(1/(1-\alpha))^{age(v)}$ be the total weighted age of the unremoved nodes and the removed nodes respectively.

The method of conditional expectations guarantees finding a \emph{good} assignment to the nodes such that
$$sb_{t,1}\geq \mathbb{E}(\Psi'_t)\geq \alpha \cdot sb_t.$$
%

We now employ an accounting method to show that each remaining node (which by Lemma~\ref{lem:manygold} is proven to be old) can charge a distinct set of $O(\Delta)$ nodes that are removed.
For the sake of the analysis, we imagine every node having a \emph{bag} containing removed nodes or fractions of removed nodes, and that the bags of all nodes are disjoint (no piece of removed node is multiply counted in the bags). Let $b_t(v)$ be the size of the bag of node $v$ at the beginning of phase $t+1$. The size of the bag corresponds to the number of removed nodes (or fractions of removed nodes) that are charged due to $v$. As mentioned, our charging is fractional in the sense that we may split a removed node $u$ into several bags $v_1, \ldots, v_k$, meaning that each $v_i$ partially blames $u$ for not being removed yet, but the total charge for a given removed node $u$ is $1$.

Recall that $\alpha$ is the bound that a golden node is removed (in the unconditional setting).
For every $t \in \{0, \ldots, \Theta(\log \Delta)\}$, we maintain the following invariants.
\begin{description}
\item{(I1)}
$b_t(v)\geq (1/(1-\alpha))^{age(v)}.$
\item{(I2)}
$\sum_{v \in V_{t,0}} (b_{t}(v)-b_{t-1}(v))=\sum_{u \in V_{t,1}} b_{t}(u)$.
\end{description}
In other words, our accounting has the property that the number of removed nodes that can be blamed for the fact that a node $v$ is not yet removed, grows exponentially with the age of $v$.

We now describe the charging process inductively and show that it maintains the invariants.
For the induction base, let $b_0(v)=1$ correspond to $v$ itself and consider the first phase in which $age(v)=0$ for every $v$. Note that  $sb_0=|V_0|=\sum_{v \in V_0} |b_0(v)|$. 

The method of conditional expectations with respect to $\Psi'_t$ guarantees that we find an assignment to the nodes such that
$sb_{0,1} \geq \alpha \cdot sb_0=\alpha \cdot |V_0|$.
The age of the unremoved nodes of $V_{0,1}$ becomes $1$ and hence to establish (I1), we need to grow their bags.

Towards this end, we empty the content of the bags of the removed nodes $V_{0,1}$ into the bags of the unremoved nodes $V_{0,0}$. In this uniform case of $t=0$, since all nodes are of the same age, it simply means that every unremoved node in $V_{0,0}$ adds into its bag an $\alpha/(1-\alpha)$ fraction of a removed node (as $sb_{0,0} \leq (1-\alpha)\cdot |V_0|$, this is doable). Hence, the size of each bag of unremoved node $v \in V_{0,0}$ becomes at least $b_1(v)\geq 1+\alpha/(1-\alpha)=1/(1-\alpha)$ and property (I1) follows. Note that since every unremoved node gets a disjoint fraction of a removed node, property (I2) follows as well. This completes the description for phase $1$.

We continue this process phase by phase, by considering phase $t$ and assuming that the invariants hold up to $t-1$.
Again, by the method of conditional expectations, we find an assignment such that
$\mathbb{E}(\Psi'_t)\geq \alpha \cdot sb_t$ and that $sb_{t,1} \geq \alpha \cdot sb_t$.
By the induction hypothesis, $b_t(v)\geq (1/(1-\alpha))^{age(v)}$ and hence
$$sb_t \leq \sum_{v \in V_t} |b_t(v)|.$$
Since the age of the unremoved nodes in $V_{t,0}$ is increased by one, their bag content should be increased. Again, this is done by emptying the content of the bags of the removed nodes $V_{t,1}$ into the bags of the unremoved nodes. Since the total amount in the bags of the removed nodes is at least $\alpha \cdot sb_t$, their content can be divided into the bags of the unremoved nodes $V_{t,0}$ such that each unit in the bag of size $(1/(1-\alpha))^{age(v)}$ is given an addition of $\alpha/(1-\alpha)$ fraction of the size of that unit, and hence
overall that new size of the bag is $(1/(1-\alpha))^{age(v)+1}$. 

Let $t^*=\Theta(\log \Delta)$. By Lemma \ref{lem:manygold}, every node that remains undecided after $t^*$ phases is of age $\Omega(\log \Delta)$. By Invariant (I1), we have that $b_{t^*}(v)\geq \Delta$ for every $v \in V_{t^*}$. 
Since the contents of $b_{t^*}(v)$ and $b_{t^*}(v')$ for every $t$ and every $v \neq v' \in V_t$ are distinct, we get that $|V_t^*|\leq n/\Delta$, and overall the total number of edges left in the graph in $O(n)$.
\end{proof}

The remaining $O(n)$ edges incident to the undecided nodes can be collected at the leader in $O(1)$ rounds using the deterministic routing algorithm of Lenzen~\cite{Lenzen13}. The leader then solves MIS for the remaining graph locally and informs the nodes. This completes the correctness of the algorithm. Theorem~\ref{theorem:MIS} follows.

\newcommand{\MISdetailed}{
Let $\mathcal{H}=\mathcal{H}_{\gamma,\beta}$ with $\gamma=\Theta(\log n)$ and $\beta=\Theta(\log \Delta)$ be given by Lemma \ref{lem: d-wise independent}.
For a hash function $h \in \mathcal{H}$ and value $p=1/2^i$ representing the probability of a node to be marked, define $m_h(v,p)=1$ if $h(ID(v)) \in [1,2^{\beta-i}]$ and $m_h(v,p)=0$ otherwise. That is, marking a node $v$ with probability $p$ is simulated deterministically by computing $m_h(v,p)$, since we output $1$ only if the value $h(ID(v))$ appears in the top $1/2^i$ fraction of the range $[1, \beta]$.
For $p=1/2^i$ and $p'=1/2^{i'}$, define
$$m_h(v,u, p,p')=
\begin{cases}
1, \mbox{~~ if ~~} h(ID(v)) \in [1,2^{\beta-i}] \mbox{~and~} h(ID(u)) \in [1, 2^{\beta-i'}].\\
0 \mbox{~~otherwise.}
\end{cases}
$$
That is $m_h(v,u, p,p')=1$ is the deterministic simulation of having both $v$ and $u$ being marked when $v,u$ are marked with probability $p,p'$.
Let $\alpha \in (0,1]$ be the constant such that every golden node is removed with probability at least $\alpha$ when using pairwise independence. Algorithm~\ref{MIS} gives the pseudocode of our deterministic MIS algorithm.

\begin{algorithm}[t]
\caption{Code for node $u$ in step $i$ of phase $t$.}
\label{MISt}
\begin{algorithmic}[1]
\Procedure{$DetMIS(t,Y_i)$} {}
\begin{flushleft}
Input: Partial graph induced on the undecided nodes and partial assignment $Y_i=(y_1=b_1, \ldots, y_i=b_i)$. \\
Output: An assignment $Y_{i+1}=(y_1=b_1, \ldots, y_i=b_i,y_{i+1}=b_{i+1})$, i.e., the goal is to extend the assignment for one more variable in the seed.
\end{flushleft}
Consider the family $\mathcal{H}_{\gamma,\beta}$ from Lemma \ref{lem: d-wise independent}, and let $\mathcal{H}(Y_i) \subseteq \mathcal{H}_{\gamma,\beta}$ be the collection of all hash functions that agree with the partial seed $Y_i$. Each function $h \in \mathcal{H}(Y_i)$ corresponds to a deterministic MIS algorithm.
\\ For every $u \in N(v) \cup \{v\}$ and $b \in \{0,1\}$, compute $m_{t,b}(u)=\sum_{h \in \mathcal{H}(Y_{i,b})}m_h(u,p_t(u))$.\\
\\ For every $u \in N(u)$, and $b \in \{0,1\}$, compute $$m_{t,b}(v,u)=\sum_{h \in \mathcal{H}(Y_{i,j})}m_h(v,u,p_t(v),p_t(u)) \mbox{~and~} M_{t,b}(v)=\sum_{u \in N(u)}m_{t,b}(v,u).$$
\\ Exchange $M_{t,0}(v),M_{t,1}(v)$ with your neighbors.
\\ If $v$ is golden node type-1:
\\ compute $\chi(v,b)=m_{t,b}(v)-M_{t,b}(v)$ for $b \in \{0,1\}$.
\\ If $v$ is golden node type-2:
\\ Define $W \subseteq N(v)$, be such that $\sum_{u \in W} p_t(u) \in [1/40,1/4]$.
\\ For every $u \in W$, and $b \in \{0,1\}$
	\\$M_{t,b}(u,W) \gets \sum_{h \in \mathcal{H}(Y_{i,j})}\sum_{w \in W \setminus \{u\}} m_h(u,w, p_t(u),p_t(w))$.
	\\$\chi(v,b) \gets \sum_{u} m_{t,b}(u)-M_{t,b}(u)- M_{t,b}(u,W)$.
If $v$ is golden:
\\  Send to the leader $(1/(1-\alpha))^{age(v)} \cdot \chi(v,b)$ for $b \in \{0,1\}$ and set $age(v)=age(v)+1$.
\\ Receive $j^*$ from the leader and set $y_{i+1}=j^*$.
\EndProcedure
\end{algorithmic}
\end{algorithm}

\begin{algorithm}[t]
\caption{Code for node $u$ in phase $t$.}
\label{MIS}
\begin{algorithmic}[1]
\Procedure{$DetMIS(t)$} {}
\begin{flushleft}
Input: Partial graph induced on the undecided nodes.
Output: Deciding if to join the MIS, be removed from the graph or remain undecided.
\end{flushleft}
\\Let $p_t(v)$ be the desired level for joining MIS, initially set $p_0=1/2$.
\\Let $d_t(v)=\sum_{u \in N(v)}p_t(u)$ be the effective degree of node $v$ in phase $t$.
\\Let $age(v)$ be the number of golden phases it had so far. Initially, $age(v)=0$.
\\Exchange the $p_t(v)$ , $d_t(v)$ and ID's with your neighbors.
\\ Set $Y_0=\emptyset, \beta=\Theta(\log n)$.
\\For $i=0, \ldots, \beta$:
\\ $(b_{i+1}) \gets DetMIS(t,Y_i)$.
\\ $Y_{i+1}=(y_1=b_1, \ldots, y_{i+1}=b_{i+1})$.
\\
\\ Let $h_t$ be selected from $\mathcal{H}_{a,\beta}$ using $Y_{\beta}$.
\\ Let $m_t(v)=m_{h_t}(v,p_t(v))$ and exchange with your neighbor.
\\ If $m_t(v)=1$ and none of your neighbors has $m_t(u)=1$, join MIS and notify your neighbors.
\\If your neighbor join MIS, removed from the graph.
\EndProcedure
\end{algorithmic}
\end{algorithm}
}

\subsection{An $O(\log \Delta)$ deterministic MIS algorithm for $\Delta=O(n^{1/3})$}
In the case where the maximal degree is bounded by  $\Delta=O(n^{1/3})$, our deterministic bounds match the randomized ones.
\begin{theorem-repeat}{theorem:MISboundedDelta}
\TheoremMISBoundedDelta
\end{theorem-repeat}

\begin{proof}
The algorithm consists of two parts as before, namely, $O(\log \Delta)$ phases that simulate the modified Ghaffari's algorithm and collecting the remaining topology at a leader and solving MIS for it locally. The second part works exactly the same as before, and so we focus on the first part which simulates the $O(\log \Delta)$ phases of the modified Ghaffari's algorithm in $O(\log \Delta)$ \emph{deterministic} rounds. The main advantage of having a small degree $\Delta=O(n^{1/3})$ is that in the congested clique, it is possible for each node to collect the entire topology of its $2$-neighborhood in $O(1)$ rounds.
This because the $2$-neighborhood of a node contains $O(\Delta^2)$ nodes, and hence there are $O(\Delta^3)=O(n)$ messages a node needs to send or receive, which can be done in $O(1)$ rounds using Lenzen's routing algorithm~\cite{Lenzen13}.

We now consider phase $t$ of the modified Ghaffari's algorithm and explain how the seed of length $O(\log n)$ can be computed in $O(1)$ rounds. Unlike the algorithm of the previous section, which computes the seed bit by bit, here the nodes compute the assignment for a \emph{chunk} of $z=\lfloor \log n \rfloor$ variables at a time.

To do so, consider the $i$'th chunk of the seed $Y'_i=(y'_1,\ldots, y'_z)$. For each of the $n$ possible assignments $(b'_1,\ldots, b'_z) \in \{0,1\}^{z}$ to the $z$ variables in $Y'$, we assign a node $u$ that receives the conditional expectation values from all the golden nodes, where the conditional expectation is computed based on assigning $y'_1=b'_1,\ldots, y'_z=b'_z$. The node $u$ then sums up all these values and obtains the expected number of removed nodes conditioned on the assignment $y'_1=b'_1,\ldots, y'_z=b'_z$. Finally, all nodes send to the leader their computed sum and the leader selects the assignment $(b^*_1,\ldots, b^*_z) \in \{0,1\}^{z}$ of largest value.

Using a well known reduction from the $(\Delta+1)$-coloring to MIS \cite{lovasz1993combinatorial,luby1986simple}\footnote{In this reduction, every node $v$ is replaced by a clique of $\Delta+1$ nodes, and we add a complete matching between cliques of neighboring nodes. When computing an MIS in this graph, there must be exactly one copy, say, $j$, for node $v$ in the MIS, which implies that the color of $v$ is $j$.}, by collecting the topology of the 2-neighborhood, we can obtain the same bounds for $(\Delta+1)$-coloring. 
\end{proof}

\subsection{An $O(D\log^2 n)$ deterministic MIS algorithm for CONGEST}
Here we provide a fast deterministic MIS algorithm for the harsher CONGEST model. For comparison, in terms of $n$ alone, the best deterministic MIS algorithm is by Panconesi and Srinivasan~\cite{panconesi1992improved} from more than 20 years ago is bounded by $2^O(\sqrt{\log n})$ rounds. While we do not improve this for the general case, we significantly improve it for graphs with a polylogarithmic diameter. The following is our main result for CONGEST.

\begin{theorem-repeat}{theorem:MIScongest}
\TheoremMISCongest
\end{theorem-repeat} 

The algorithm is very similar to Theorem. \ref{theorem:MIS} with two main differences. First, we run Ghaffari's algorithm for $O(\log n)$ rounds instead of $O(\log \Delta)$ rounds.
Each round is simulated by a phase with $O(D \log n)$ rounds. Specifically, in each phase, we need to compute the seed of length $O(\log n)$, this is done bit by bit using the method of conditional expectations and aggregating the result at some leader node. The leader then notifies the assignment of the bit to the entire graph. Since, each bit in the seed is computed in $O(D)$ rounds, overall the run time is  $O(D \log^2 n)$. 

For the correctness, assume towards contradiction that after $\Omega(\log n)$ rounds, at least one node remains undecided. Then, by the proof Lemma \ref{lem:amortized}, every node that survives, can charge $\Omega(n^{c})$ nodes that are removed, contradiction as there only $n$ nodes. 
%
%
%
%
%
%
%
\section{Deterministic spanner construction}
In this section we present a derandomization algorithm in the congested clique for the spanner construction of Baswana-Sen~\cite{BaswanaS07}. In general, we use the same general outline as in the MIS derandomization: We first reduce the dependence between the coins used by the algorithm and then use the method of conditional expectations for every iteration of the algorithm. However, here we face different challenges that we need to overcome.

The following is the main theorem of this section.

\begin{theorem-repeat}{theorem:Spanner}
\SpannerTheorem
\end{theorem-repeat}

We first present the original algorithm of Baswana-Sen~\cite{BaswanaS07}, which constructs a $(2k-1)$-spanner with $O(kn^{1+1/k})$ edges in $O(k^2)$ rounds. Next, we consider the same algorithm with only limited independence between its coin tosses. We prove some properties of the algorithm and show it can be derandomized. Finally we present our deterministic algorithm for the congested clique which constructs a $(2k-1)$-spanner with $O(kn^{1+1/k}\log n)$ edges within $O(k\log n)$ rounds.

\paragraph{The randomized spanner algorithm.}
We begin by presenting a simplified version of the Baswana-Sen algorithm. For the full details of the Baswana-Sen algorithm we refer the reader to \cite{BaswanaS07}. 


At each iteration the algorithm maintains a \emph{clustering} of the nodes. A \emph{cluster} is a subset of nodes, and a clustering is a set of disjoint clusters. In the distributed setting each cluster has a leader, and a spanning tree rooted at the leader is maintained inside the cluster. We will abuse notation and say that a cluster performs a certain action. When we say this, it means that the leader gathers the required information from the cluster nodes to make a decision, and propagates relevant data down the cluster tree. We will also refer at times to the ID of a cluster, which is the ID of the cluster leader.

We denote the clustering maintained at iteration $i$ by $\mathcal{C}_i$, where initially $\mathcal{C}_0 = \set{\set{v} \mid v \in V}$. At each iteration, $\mathcal{C}_i$ is sampled from $\mathcal{C}_{i-1}$, by having every cluster in $\mathcal{C}_{i-1}$ join $\mathcal{C}_i$ with probability $n^{-1/k}$. In the final iteration we force $\mathcal{C}_k=\emptyset$.
A node $v$ that belongs to a cluster $C \in \mathcal{C}_i$ is called $i$-\emph{clustered}, and otherwise it is $i$-\emph{unlcustered}. 

The algorithm also maintains a set of edges, $E'$, initialized to $E$. For every edge $(u,v)$ removed from $E'$ during the algorithm, it is guaranteed that there is a path from $u$ to $v$ in the constructed spanner, $H$, consisting of at most $2k-1$ edges, each one of weight not greater than the weight of $(u,v)$.

Let $v \in V$ be a node that stops being clustered at iteration $i$, and let $E'(v,C)=\set{(v,u) \in E' \mid u\in C}$ be the set of edges between a node and a cluster $C$, for every $C \in \mathcal{C}_i$. Let $e_{v,C}$ be the lightest edge in $E'(v,C)$. 
Let $L$ be the set of lightest edges between $v$ and the clusters in $\mathcal{C}_{i-1}$. We go over $L$ in ascending order of edge weights and add every edge connecting $v$ to cluster $C$ and then we discard $E'(v,C)$ from $E'$. We say that $v$ adds these edges at iteration $i$. If we reach a cluster $C \in \mathcal{C}_i$, we continue to the next node. 
The pseudocode appears in Algorithms~\ref{alg-spanner-rand} and~\ref{alg-spanner-rand-iteration}.

\RestyleAlgo{boxruled}
\LinesNumbered
\begin{algorithm}[htbp]
	\caption{Randomized $(2k-1)$-spanner construction}
\label{alg-spanner-rand}

	$H= \emptyset$\\
	$\mathcal{C}_0 = \set{\set{v} \mid v \in V}$\\
	$E' = E$\\
	
	\For{$i$ from 1 to $k$}
	{
		\If{$i=k$}{
			$\mathcal{C}_k=\emptyset$
		}
		\Else
		{
			$\mathcal{C}_i$ is sampled from $\mathcal{C}_{i-1}$ by sampling each cluster with probability $n^{-1/k}$
		}
		Run Algorithm~\ref{alg-spanner-rand-iteration}
	}
	
\end{algorithm}
\RestyleAlgo{boxruled}
\LinesNumbered
\begin{algorithm}[htbp]
	\caption{Baswana-Sen iteration}
\label{alg-spanner-rand-iteration}
	
	\ForEach{node $v$ that stopped being clustered at iteration $i$}
	{
		$L=\set{e_{v,C} \mid C \in \mathcal{C}_{i-1}}$\\
		Let $e_j$ be the $j$-th edge in $L$ in ascending weight\\
		\For{$j=1$ to $j=\ell$}
		{
			
			$H = H\cup \set{e_j}$\\
			Let $C$ be the cluster corresponding to $e_j$\\
			$E' = E' \setminus E'(v,C)$\\
			\If{$C \in \mathcal{C}_i$}
			{
				break	
			} 	
		}
	}
	
\end{algorithm}

Algorithm~\ref{alg-spanner-rand} is guaranteed to finish after $O(k^2)$ communication rounds in the distributed setting, and return a $(2k-1)$-spanner of expected size $O(kn^{1+1/k})$.

%

\paragraph{$d$-wise independence and derandomization.} The Baswana-Sen algorithm does not work as is with reduced independence, because the bound on the spanner size relies on full independence between the coin flips of the clusters.
However, we proceed by establishing properties of the Basawana-Sen algorithm (Algorithm~\ref{alg-spanner-rand}) that do hold in the case of limited independence between its coin tosses. We use following result of Benjamini et al.~\cite{benjamini2012k}.

\begin{theorem}
\label{theorem:Mndp}
	Let $M(n,d,p)$ be the maximal probability of the AND event for $n$ binary $d$-wise independent random variables, each with probability $p$ of having the value $1$. If $d$ is even, then:
	$$M(n,d,p) \leq \frac{p^n}{\Pr[Bin(n,1-p) \leq d/2]},$$
	and if $d$ is odd, then:
	$$M(n,d,p) = pM(n-1,d-1,p).$$
\end{theorem}
We also use the following Chernoff bound for $d$-wise independent random variables from \cite{SchmidtSS95}.

\begin{theorem}
	\label{thm:d-wise chernoff}
	Let $X_1,...,X_n$ be $d$-wise independent random variables taking values in $[0,1]$, where $X = \sum_{i=1}^n X_i$ and $\E[X]=\mu$. Then for all $\epsilon \leq 1$ we have that if $d \leq \lfloor \epsilon^2 \mu e^{-1/3} \rfloor$ then:
	$$\Pr[|X- \mu | \geq \epsilon \mu] \leq e^{-\lfloor d/2 \rfloor}.$$
	And if $d > \lfloor \epsilon^2 \mu e^{-1/3} \rfloor$ then:
	$$\Pr[|X- \mu | \geq \epsilon \mu] \leq e^{-\lfloor \epsilon^2 \mu /3 \rfloor}.$$
	
\end{theorem}


We consider Algorithm~\ref{alg-spanner-rand} with only $O(\log n)$-wise independence between the coin tosses of clusters at each iteration. We also assume that $k \leq 0.5\log n $.
We prove the following two lemmas that will be used later for derandomization.

Let $d=2\log 2n$ be the independence parameter, and define $\xi = e^{1/3}2\log 2n$, and $\alpha_i = \prod_{j=1}^i (1+1/(k-j))$.

\begin{lemma}
	\label{lem: cluster size prob}
For every $1 \leq i \leq k-1$, if $|\mathcal{C}_{i-1}| \leq \xi \alpha_{i-1}n^{1-(i-1)/k } $, then $\Pr[|\mathcal{C}_i| \geq \xi \alpha_i n^{1-i/k}] <  0.5$. In addition, $\xi \alpha_{k-1}n^{1/k} = O(kn^{1/k} \log n)$.
\end{lemma}
\begin{proof}
	We define for every cluster $C$ the indicator random variable $X(C)$ for the event that the cluster remains for the next iteration. Note that $|\mathcal{C}_i|=\sum X(C)$ and $\E[X(C)]= n^{-1/k}$. By the assumption of the lemma, for the $(i-1)$-th iteration we know that we have at most $\xi \alpha_{i-1}n^{1-(i-1)/k} $ clusters left from the previous iteration. Thus $\E[\sum X(C)] \leq \xi \alpha_{i-1}n^{1-(i-1)/k}  \cdot n^{-1/k} \leq \xi \alpha_{i-1}n^{1-i/k} $.
	
	We wish to apply Theorem~\ref{thm:d-wise chernoff} with $d=2\log 2n, \mu_i = \xi \alpha_{i-1}n^{1-i/k}$ and $\epsilon_i=1/(k-i)$. We note that $\alpha_i \geq 1$ for every $i$.
	We now show that it is always the case that $d \leq \lfloor \epsilon_{i}^{2} \mu_i e^{-1/3} \rfloor$, so we can use the first case of Theorem~\ref{thm:d-wise chernoff}.
	Plugging in $\epsilon_i,\mu_i,d$ gives that we need to prove that: 
	$$2\log 2n \leq  2\log (2n) \cdot e^{1/3}e^{-1/3} \alpha_{i-1} n^{1-i/k} /(k-i)^2,$$
	which holds if and only if  
	$$\alpha_{i-1} n^{1-i/k} /(k-i)^2 \geq 1 . $$
	We bound the left hand side from below by
	$$\alpha_{i-1} n^{1-i/k} /(k-i)^2 \geq  n^{1-i/k} /(k-i)^2.$$
	To prove that the above is at least $1$, we claim that (a) the function $n^{1-i/k} /(k-i)^2$ is monotonically decreasing for $1 \leq i \leq k - 1$, and (b) that  $ n^{1-i/k} /(k-i)^2 \geq 1 $ when $i=k-1$. To prove (a), we prove that $n^{1-i/k} /(k-i)^2 \leq n^{1-(i-1)/k} /(k-(i-1))^2$. Taking the square root of both sides gives that we need to prove that
		$$n^{1/2-i/2k} /(k-i) \leq \\ n^{1/2-(i-1)/2k} /(k-(i-1)),$$
		which holds if and only if
		 $$(k-(i-1)) / (k-i) \leq n^{1/2k}.$$
	For the left hand side of the above, it holds that
	$$ (k-(i-1))/(k-i) \leq 1 + 1/(k-i) \leq 2,$$
	and since we assumed that $k < 0.5 \log n$, we have that $n^{1/2k} \geq n^{1/\log n} = 2$. Therefore, $n^{1/2k}  \geq (k-(i-1))/(k-i)$ as required for (a). 
	
	We now show (b), that is, that $ n^{1-i/k} /(k-i)^2 \geq 1 $ when $i=k-1$. This holds since for $i=k-1$ we have $ n^{1-i/k} /(k-i)^2 = n^{1-(k-1)/k} / (k-(k-1))^2 = n^{1/k} \geq 1$, giving (b). This establishes that $d \leq \lfloor \epsilon_{i}^{2} \mu_i e^{-1/3} \rfloor$, and thus the first condition of Theorem~\ref{thm:d-wise chernoff} always holds.
	


	Since $\alpha_i = (1+1/(k-i))\alpha_{i-1}=(1+\epsilon_i )\alpha_{i-1}$ we have 
	$$\Pr[|\mathcal{C}_i| \geq \xi \alpha_{i} n^{1-i/k}] = \Pr[\sum X(C) \geq \xi (1+\epsilon_i)\alpha_{i-1} n^{1-i/k}] = \Pr[\sum X(C) \geq (1+ \epsilon_i ) \mu_i] .$$
	We now apply Theorem~\ref{thm:d-wise chernoff} and obtain
	$$\Pr[\sum X(C) - \mu_i \geq  \epsilon_i \mu_i] < e^{-\lfloor d/2 \rfloor}  < 0.5,
	$$
	which proves the first part of the lemma, that $\Pr[|\mathcal{C}_i| \geq \xi \alpha_{i} n^{1-i/k}]  < 0.5$.
	
	Finally, we have
	$$ \alpha_{k-1} = \prod_{j=1}^{k-1} (1+1/(k-j)) \leq e^{\sum_{j=1}^i 1/(k-j)} = e^{\sum_{j=1}^{k-1} 1/j} = O(k).$$
	Which implies the second part of the lemma, that $\xi \alpha_{k-1}n^{1/k} = O(kn^{1/k} \log n)$, and completes the proof.
\end{proof}

Fix an iteration $i$ and consider an $i$-unclustered node $v$. Denote by $X_v$ the indicator variable for the event that node $v$ adds more than $t=2 n^{1/k} \log n$ edges in this iteration.
\begin{lemma}
	\label{lem: edges added prob}
	The probability that there exists a node $v$ at some iteration which adds more than $t$ edges to the spanner is less than 0.5. Formally, $\Pr[\vee X_v=1] <0.5$.
\end{lemma}

\begin{proof}
	From the union bound it holds that $\Pr[\vee X_v=1] \leq \Pr[X_v=1]$. Next,  we bound $\sum \Pr[X_v=1]$. 
	We show that every $\Pr[X_v=1]$ is smaller than $1/2n$, completing the proof by applying a union bound over all nodes.
	Let $\ell$ be the number of neighboring clusters of $v$ in $\mathcal{C}_{i-1}$. If $\ell \leq t$ then $\Pr[X_v=1]=0$. Otherwise, we might add $t$ edges to $H$, if and only if the clusters corresponding to the $t$ lightest edges in $L$ are not in $\mathcal{C}_i$. This is the value $M(t,2d,p)$ (we use $2d$ to avoid fractions in the binomial coefficient) with $p=1-n^{-1/k}$. Let us bound $M(t,2d,p)$ as follows.
	\begin{align*}
		M(t,2d,p) &\leq \frac{p^t}{\Pr[Bin(t,1-p) \leq d]} \leq \frac{p^t}{\binom{t}{d}(1-p)^d p^{t-d}} = \frac{p^d}{\binom{t}{d}(1-p)^d}  \\
		& \leq \frac{1}{\binom{t}{d}(1-p)^d} \leq \frac{1}{(t/d)^d(1-p)^d} \leq \frac{d^d}{t^d(1-p)^d}.
	\end{align*}
	Plugging in $p=1-n^{-1/k}$ and $t=2 n^{1/k} \log n$ gives	
	$$M(2n^{1/k}\log n,2d,1-n^{-1/k}) \leq \frac{d^d}{(2n^{1/k}\log n)^d(n^{-1/k})^d} = \frac{d^d}{(2\log n)^d}.$$
	Now let us plug in $d=2\log 2n$ and we get:
	
	$$M(2n^{1/k}\log n,2\log 2n,1-n^{-1/k}) \leq (1/2)^{2\log 2n} < 1/2n. $$
	Finally, as explained, we use a union bound to get that $\Pr[\vee X_v=1] \leq \sum \Pr[X_v=1]<0.5$.
\end{proof}

The above lemmas do not guarantee that the algorithm yields the same expected spanner size as the algorithm with full independence, but using these lemmas we can now construct a deterministic algorithm.

Let us define two bad events that can occur during some iteration $i$ of the algorithm. Let $A$ be the event that not enough clusters were removed during the iteration, and let $B$ be the event that there exists a node that adds too many edges to the spanner. We will define these events formally later on. Let $X_A, X_B$ be the corresponding indicator random variables for the events. Assume that it holds that $\E[X_A]+\E[X_B] < 1$. In this case we can use the method of conditional expectations in order to get an assignment to our random coins such that no bad event occurs.

Let $\rhobar$ be the vector of coin flips used by the clusters. Let $\ybar$ be the seed randomness from Lemma~\ref{lem: d-wise independent} used to generate $\rhobar$ such that its entries are $d$-wise independent, where $d=O(\log n)$. We use $\ybar$ to select a function $h\in \Hcal_{\gamma,\beta}$, where $\gamma=\log n$ and $\beta = \log n^{1/k}$. Each node $v$ uses $\ybar$ to generate $h$ and then uses the value $h(ID(v))$ to generate $\rhobar[v]$.

Let $Z= (z_1, \dots, z_n)$ be the final assignment generated by the method of conditional expectations. Then, $\E[X_A \mid \ybar=Z]+\E[X_B \mid \ybar=Z] < 1$. Because $X_A$ and $X_B$ are binary variables that are functions of $\ybar$, it must be the case that both are zero.
We can write our expectation as follows:
\begin{align*}
& \E[X_A]+\E[X_B] = \Pr[X_A=1] + \Pr[X_B=1] = \Pr[X_A=1] + \Pr[\vee X_v=1] 
\end{align*}

At every iteration of the algorithm we would like to keep $\E[X_A]+\E[X_B]$ below 1, which would guarantee both bad events do not occur. Unfortunately, it is unclear how to compute $\Pr[\vee X_v=1]$ conditioned on some assignment to $\ybar$. Thus, we must use a pessimistic estimator. We consider $\sum_v \Pr[X_v=1]$, and we have that:
\begin{align*}
\Pr[X_A=1] + \Pr[\vee X_v=1] \leq \Pr[X_A=1] + \sum_v \Pr[X_v=1].
\end{align*}
We define our pessimistic estimator $\Psi = X_A + \sum_v X_v$.
Note that the above inequality holds conditioned on any partial assignment to $\ybar$, because it is derived via a union bound. Thus, if we show that $\E[\Psi] = \Pr[X_A=1] + \sum \Pr[X_v=1] < 1$, it is enough to apply the method of conditional expectations for $\Psi$, keeping the expression below 1. For the assignment $Z$ resulting from this process it will hold that $\E[X_A \mid \ybar=Z]+\E[X_B \mid \ybar = Z] < \E[\Psi \mid \ybar = Z] < 1$, as required.

It remains only to bound the pessimistic estimator $\Psi$. 
This can be achieved using Lemma~\ref{lem: cluster size prob} and Lemma~\ref{lem: edges added prob}. In each iteration of the algorithm, because the bad event $A$ did not occur in the previous iteration, the condition that $|\mathcal{C}_{i-1}| \leq \xi \alpha_{i-1} n^{1-(i-1)/k}$ holds for Lemma~\ref{lem: edges added prob}. This yields $\Pr[X_A=1] + \sum \Pr[X_v=1]<1$.

\paragraph{The deterministic spanner construction in the congested clique.}
We are now ready to describe our algorithm in the congested clique. We first show the we can indeed compute the conditional expectation of our pessimistic estimator $\Psi$. 
We are interested in $\Pr[X_A=1 \mid y_1=b_1,\dots, y_i=b_i]$ and $\Pr[X_v=1 \mid y_1=b_1,\dots, y_i=b_i]$. 
 Knowing some partial assignment to $\ybar$, we can iterate over all possible selections of $h \in \Hcal_{\gamma, \beta}$ and compute the coin flip for every cluster using its ID alone.
The expression $\Pr[X_A=1 \mid y_1=b_1,\dots, y_i=b_i]$ is just the probability of enough clusters getting removed given some partial assignment. It does not depend on the graph topology, and can easily be computed by a node only knowing the IDs of clusters currently active. To compute $\Pr[X_v=1 \mid y_1=b_1,\dots, y_i=b_i]$ the node $v$ can collect all of the IDs from neighboring clusters and go over all possibilities for calculate the probability of adding too many edges.


\RestyleAlgo{boxruled}
\LinesNumbered
\begin{algorithm}[htbp]
	\caption{deterministic $(2k-1)$-spanner algorithm}
	\label{alg-spanner-det}

	$H= \emptyset$\\
	$\mathcal{C}_0 = \set{\set{v} \mid v \in V}$\\
	$E'=E$
	
	\For{$i$ from 1 to $k$}
	{
		$\phi = \emptyset$ //partial assignment\\
		\ForEach{$v \in V$ simultaneously}
		{
			\If{$v$ is cluster leader for $C \in \mathcal{C}_{i-1}$}
			{
				Send $ID(v)$ to all nodes
			}
			\For{$j \in [(\log n)]$}{
				\For{$\tau \in [\log n]$}{
					Compute $x_\tau=\Pr[X_v \mid \phi,y_j=\tau]$\\
					Send $(x_\tau, \tau)$ to $u\in V, ID(u) = \tau$\\
				}
				
				\If{$ID(v)=\tau, \tau \in [\log n]$ }
				{
					
					$s = \Pr[X_A \mid \phi, y_j=\tau] + \sum_{(x,\tau)} x $\\
					Send $(\tau,s)$ to main leader\\
				}
				\If{leader}
				{
					$\tau_{min} = argmin_\tau \{ s \mid (\tau,s) \}$\\
					$\phi = \phi \cup \set{y_j = \tau_{min}}$\\
					send updated $\phi$ to all nodes
				}
			}
		}
		\If{$i=k$}{
			$\mathcal{C}_i=\emptyset$
		}
		\Else
		{
			$\mathcal{C}_i$ sampled from $\mathcal{C}_{i-1}$ using $\ybar$
		}
		Run Algorithm~\ref{alg-spanner-rand-iteration}\\
		
	}	
	
\end{algorithm}

Algorithm~\ref{alg-spanner-det} is the pseudocode, where before running an iteration of the Baswana-Sen algorithm we first find a seed randomness $\ybar$, such that both bad events $A$ and $B$ do not occur. We then execute an iteration of the Baswana-Sen algorithm using the seed to assign a random coin for each cluster. Because neither of the bad events occur in any iteration, no node adds more than $2n^{1/k}\log n$ edges in any iteration, and we reach the final iteration with $O(n^{1/k})$ clusters. Therefore, each iteration adds no more than $O(n^{1+1/k}\log n)$ edges, and the final iteration adds no more than $O(kn^{1+1/k}\log n)$ edges (assuming a loose bound of having all nodes connect to all remaining clusters).
We conclude that our spanner has $O(kn^{1+1/k} \log n)$ edges.

We find $\ybar$ via the method of conditional expectations, keeping the pessimistic estimator below 1. We consider the value of the pessimistic estimator under some partial assignment to $\ybar$, and extend the assignment such that the pessimistic estimator is kept below 1. 

When finding $\ybar$ we bring the power of the congested clique to our aid. The sequential approach would go over $\ybar$ bit by bit, setting it to the value which optimizes the pessimistic estimator until all values of $\ybar$ are fully set. In the congested cliques we can go over blocks of $\ybar$ of size $\log n$, calculating the value of the pessimistic estimator for each one of the $n$ possible assignments of the block. We achieve this by assigning each node to be responsible for aggregating the data in order to calculate the pessimistic estimator for one of the possible $n$ values. This speeds up our calculation by a $\log n$ factor.

The above is implemented in the algorithm as follows: Each node $v\in V$ iterates over all $\log n$ blocks of $\ybar$, each of size $\log n$. For each block it computes $\Pr[X_v]$ conditioned on all $n$ values of the block. For every value $\tau$ of the block it sends each the conditional probability to $u_\tau$ which is responsible for computing the value of the pessimistic estimator conditioned on the value $\tau$ for the block. Knowing the conditional value of $\Pr[X_v]$ for every $v\in V$ and the IDs of the active clusters, the node $u_\tau$ can now compute the value of the conditional pessimistic estimator. All of the conditional values of the pessimistic estimator are then aggregated to a leader node which picks the value that minimizes the pessimistic estimator. Finally, the leader broadcasts the selected value for the block to all nodes. All nodes then continue to the next iteration. After computing $\ybar$ we run an iteration of Baswana-Sen where the coin tosses of clusters are generated from $\ybar$.

Another benefit of running the Baswana-Sen algorithm in the congested clique is that we save an $O(k)$ factor in our round complexity. This is because cluster nodes may now communicate with the cluster leader directly, instead of propagating their message via other cluster nodes. This takes $O(k)$ in the standard distributed setting because the distance to the center of each cluster is at most the iteration number.

We conclude that the round complexity of our algorithm is the number of iterations of the Baswana-Sen main loop in the congested clique, which is $O(k)$, multiplied by the overhead of guaranteeing the bad events $A,B$ will not happen during the iteration. We guarantee this by applying the method of conditional expectation over $\ybar$, using a block of size $\log n$ at each step of the method of conditional expectations.

\sloppy{
We note that each cluster flips a biased coin with probability $n^{-1/k}$, and we require $d$-wise independence between the coin flips. We conclude from Lemma~\ref{lem: d-wise independent} that the size of $\ybar$ is $O(d \max \set{\log n^{1/k}, \log n}) = O(\log^2 n)$ bits. Because, we pay $O(1)$ rounds for every $\log n$ chunk of $\ybar$, we conclude from the above that our algorithm takes a total of $O(k\log n)$ communication rounds. This completes the proof of Theorem~\ref{theorem:Spanner}.
}

\section{Discussion}

The main research question this paper addresses is the deterministic complexity of so-called local problems under bandwidth restrictions. Specifically, we derandomize an MIS algorithm and a spanner construction in the congested clique model, and derandomize an MIS algorithm in the CONGEST model. This greatly improves upon the previously known results. Whereas our techniques imply that many local algorithms can be derandomized in the congested-clique (e.g., hitting set, ruling sets, coloring, matching etc.), the situation appears to be fundamentally different for global task such as connectivity, min-cut and MST. For instance, the best randomized MST algorithm in the congested-clique has time complexity of $O(\log^*n)$ rounds \cite{Ghaffari-Parter-MST}, but the best deterministic bound is $O(\log \log n)$ rounds\cite{lotker2003mst}.  Derandomization of such global tasks might require different techniques. 

The importance of randomness in \emph{local} computation lies in the fact that recent developments~\cite{ChangKP16} show separations between randomized and deterministic complexities in the unlimited bandwidth setting of the LOCAL model. While some distributed algorithms happen to use small messages, our understanding of the impact of message size on the complexity of local problems is in its infancy. 

This work opens a window to many additional intriguing questions. First, we would like to see many more local problems being derandomized despite congestion restrictions. Alternatively, significant progress would be made by otherwise devising deterministic algorithms for this setting. Finally, understanding the relative power of randomization with bandwidth restrictions is a worthy aim for future research.

\paragraph{Acknowledgments:} We are very grateful to Mohsen Ghaffari for many helpful discussions and useful observations involving the derandomization of his MIS algorithm.

\bibliographystyle{alpha}
\bibliography{derand_main}

\appendix
\section{Pseudocode for the Deterministic MIS algorithm}
\label{sec:code}
\label{sec:MISdetailed}

Let $\mathcal{H}=\mathcal{H}_{\gamma,\beta}$ with $\gamma=\Theta(\log n)$ and $\beta=\Theta(\log \Delta)$ be given by Lemma \ref{lem: d-wise independent}. Let $\mathcal{H}(Y_i) \subseteq \mathcal{H}_{\gamma,\beta}$ be the collection of all hash functions that agree with the partial seed $Y_i$. Each function $h \in \mathcal{H}(Y_i)$ corresponds to a deterministic MIS algorithm.

For a hash function $h \in \mathcal{H}$ and value $p=1/2^i$ representing the probability of a node to be marked, define $m_h(v,p)=1$ if $h(ID(v)) \in [1,2^{\beta-i}]$ and $m_h(v,p)=0$ otherwise. That is, marking a node $v$ with probability $p$ is simulated deterministically by computing $m_h(v,p)$, since we output $1$ only if the value $h(ID(v))$ appears in the top $1/2^i$ fraction of the range $[1, \beta]$.
For $p=1/2^i$ and $p'=1/2^{i'}$, define
$$m_h(v,u, p,p')=
\begin{cases}
1, \mbox{~~ if ~~} h(ID(v)) \in [1,2^{\beta-i}] \mbox{~and~} h(ID(u)) \in [1, 2^{\beta-i'}].\\
0,  \mbox{~~otherwise.}
\end{cases}
$$
That is $m_h(v,u, p,p')=1$ is the deterministic simulation of having both $v$ and $u$ being marked when $v,u$ are marked with probability $p,p'$.
Let $\alpha \in (0,1]$ be the constant such that every golden node is removed with probability at least $\alpha$ when using pairwise independence. Algorithm~\ref{MIS} gives the pseudocode of our deterministic MIS algorithm.

\RestyleAlgo{boxruled}
\LinesNumbered
\begin{algorithm}[htbp]
	\caption{$\texttt{DetMIS}(t,Y_i)$: Code for node $v$ in step $i$ of phase $t$.}
\label{MISt}
Input: \\
\quad\quad A partial graph induced on the undecided nodes\\
\quad\quad A partial assignment $Y_i=(y_1=b_1, \ldots, y_i=b_i)$\\
Output: \\
\quad\quad An assignment $Y_{i+1}=(y_1=b_1, \ldots, y_i=b_i,y_{i+1}=b_{i+1})$\\
\quad\quad // the goal is to extend the assignment for one more variable in the seed

\For{$u \in N(v) \cup \{v\}$ and $b \in \{0,1\}$}
{
	$m_{t,b}(u) \gets \sum_{h \in \mathcal{H}(Y_{i,b})}m_h(u,p_t(u))$
}
\For{$u \in N(v)$, and $b \in \{0,1\}$}
{
	$m_{t,b}(v,u) \gets \sum_{h \in \mathcal{H}(Y_{i,j})}m_h(v,u,p_t(v),p_t(u))$ 
}
$M_{t,b}(v)\gets \sum_{u \in N(u)}m_{t,b}(v,u)$\\
Exchange $M_{t,0}(v),M_{t,1}(v)$ with neighbors\\
\If {$v$ is a golden type-1 node}
{
	$\chi(v,b) \gets m_{t,b}(v)-M_{t,b}(v)$ for $b \in \{0,1\}$\\
	// $\chi(v,b)$ corresponds to $\mathbb{E}(\psi_{v,t} ~\mid~ Y_{i,b})$
}
\If{$v$ is a golden type-2 node}
{
	Define $W(v) \subseteq N(v)$ such that $\sum_{u \in W(v)} p_t(u) \in [1/40,1/4]$\\
	\For{$u \in W(v)$, and $b \in \{0,1\}$}
	{
		$M_{t,b}(u,W(v)) \gets \sum_{h \in \mathcal{H}(Y_{i,j})}\sum_{w \in W(v) \setminus \{u\}} m_h(u,w, p_t(u),p_t(w))$
	}
	$\chi(v,b) \gets \sum_{u} m_{t,b}(u)-M_{t,b}(u)- M_{t,b}(u,W(v))$\\
	// $\chi(v,b)$ corresponds to $\mathbb{E}(\psi_{v,t} ~\mid~ Y_{i,b})$
}
\If{$v$ is golden}
{
	Send to the leader $(1/(1-\alpha))^{age(v)} \cdot \chi(v,b)$ for $b \in \{0,1\}$\\
	// All of the above fits in an $O(\log n)$-bit message \\
	$age(v) \gets age(v)+1$
}
Receive $j^*$ from the leader\\
$y_{i+1} \gets j^*$
\end{algorithm}

\RestyleAlgo{boxruled}
\LinesNumbered
\begin{algorithm}[htbp]
	\caption{$\texttt{DetMIS}(t)$: Code for node $v$ in phase $t$.}
\label{MIS}
Input:\\
\quad A partial graph induced on the undecided nodes\\
Output:\\
\quad Decide whether to join the MIS, be removed from the graph, or remain undecided

Let $p_t(v)$ be the desired level for joining MIS, initially $p_0(v) \gets 1/2$\\
Let $d_t(v)=\sum_{u \in N(v)}p_t(u)$ be the effective degree of node $v$ in phase $t$\\
Let $age(v)$ be the number of golden phases it had so far. Initially, $age(v) \gets 0$\\
Exchange the $p_t(v)$ , $d_t(v)$ and ID with neighbors\\
$Y_0=\gets \emptyset, \beta\gets \Theta(\log n)$\\
\For{$i=0, \ldots, \beta$}
{
	$(b_{i+1}) \gets \texttt{DetMIS}(t,Y_i)$\\
	$Y_{i+1}\gets (y_1=b_1, \ldots, y_{i+1}=b_{i+1})$
}
Let $h_t$ be selected from $\mathcal{H}_{a,\beta}$ using $Y_{\beta}$\\
Let $m_t(v) \gets m_{h_t}(v,p_t(v))$\\
Exchange $m_t(v)$ with neighbor\\
\If{$m_t(v)=1$ and none of your neighbors has $m_t(u)=1$}
{
	join MIS\\
	notify neighbors
}
\If{$\exists$ neighbor that joins MIS}
{
	notify neighbors of being removed from graph
}
\end{algorithm}

%

\end{document}